\documentclass[12pt, a4paper, fleqn]{article}

\usepackage{etex}
\usepackage[utf8]{inputenc}
\usepackage[TS1,T1]{fontenc}
\usepackage{xspace}
\usepackage{array}
\newcolumntype{P}[1]{>{\centering\arraybackslash}p{#1}}
\usepackage{multirow}
\usepackage{tikz-cd}
\usetikzlibrary{external}
\usepackage{smartdiagram}
 \usepackage{tensor,tabularx}
\usetikzlibrary{decorations.markings}
\usepackage{layout} 
\usetikzlibrary{arrows.meta} 
\usepackage{pdfpages}
\usepackage{float}
\usepackage[a4paper,tmargin=2.75truecm,bmargin=2.5truecm,hmargin=2.8truecm,  rmargin=2.9truecm, lmargin=2.9truecm, marginparwidth=25mm]{geometry}

\usetikzlibrary{arrows,matrix,positioning,shapes,arrows}
\usetikzlibrary{shapes.geometric, arrows, calc, intersections}
\newcommand{\tikznode}[2]{\relax
	\ifmmode%
	\tikz[remember picture,baseline=(#1.base),inner sep=0pt] \node (#1) {$#2$};
	\else
	\tikz[remember picture,baseline=(#1.base),inner sep=0pt] \node (#1) {#2};%
	\fi}

\usepackage[english]{babel}
\usepackage{blindtext} 
\usepackage{epigraph} 
\usepackage[]{libertinus-type1}
\usepackage{amsfonts, mathtools,slashed,mathdots,bm}
\usepackage{libertinust1math}
\usepackage[bb=ams]{mathalpha}
\makeatletter
\newcommand{\addbar@}[3]{%
	\makebox[0pt][l]{%
		\raisebox{#1}[0pt][0pt]{%
			\kern#2
			\scalebox{#3}[0.8]{$\m@th\mathchar"84$}%
		}%
	}%
}
\DeclareRobustCommand{\lambdabar}{\text{\addbar@{0.1ex}{0.18em}{1}}\lambda}
\makeatother
\usepackage{slashed}
\usepackage{enumitem}
\newlist{enum-hypothesis}{enumerate}{1}
\setlist[enum-hypothesis]{label=(\arabic*),itemsep=0pt, parsep=0pt}

\setlist[enumerate,1]{label=\arabic*., ref=\arabic*, topsep=1pt, itemsep=2pt, parsep=0pt, leftmargin=1.5em, itemindent=0em, labelsep=0.2em, labelwidth=1.3em}
\setlist[enumerate,2]{label=\alph*., ref=\theenumi.\alph*, topsep=1pt, itemsep=2pt, parsep=0pt, leftmargin=0.5em, itemindent=0em, labelsep=0.2em, labelwidth=1.5em}
\setlist[enumerate,3]{label=\roman*., ref=\theenumii.\roman*, topsep=1pt, itemsep=2pt, parsep=0pt, leftmargin=0.5em, itemindent=0em, labelsep=0.2em, labelwidth=1.2em}

\usepackage{booktabs, dcolumn}
\usepackage{graphicx, subfig}

\usepackage{pbox}
 
\usepackage[square,numbers,compress,merge]{natbib}
\usepackage[thmmarks,amsmath]{ntheorem}

\newtheorem{theorem}{Theorem}[section]
\newtheorem{proposition}[theorem]{Proposition}
\newtheorem{lemma}[theorem]{Lemma}
\newtheorem{corollary}[theorem]{Corollary}

\newtheorem{definition}[theorem]{Definition}

\theoremsymbol{$\square$}
\theoremstyle{plain}
\theorembodyfont{\upshape}
\newtheorem{remark}[theorem]{Remark}
\theoremsymbol{$\Diamond$}

\theoremsymbol{}
\theoremstyle{break}
\theorembodyfont{\slshape}

\theoremheaderfont{\scshape}
\theorembodyfont{\upshape} 
\theoremstyle{nonumberplain}
\theoremseparator{} 
\theoremsymbol{\rule{1.3ex}{1.3ex}} 
\newtheorem{proof}{Proof}

\usepackage{hyperref}

\hypersetup{
plainpages=false,
colorlinks=true,
linkcolor=blue, 
anchorcolor=black, 
citecolor=blue, 
urlcolor=blue, 
menucolor=black, 
filecolor=black, 
bookmarksopen=true,
bookmarksnumbered=true}

\hypersetup{
pdfauthor={Gaston Nieuviarts},
pdftitle={Titre},
pdfsubject={},
pdfcreator={pdflatex},
pdfproducer={pdflatex},
pdfkeywords={}
}

\newcommand{\bbR}{\mathbb{R}}

\newcommand\bbZ{\mathbb{Z}}

\newcommand{\Man}{\mathcal{M}}
 
\newcommand{\bbbone}{{\text{\usefont{U}{bbold}{m}{n}\char49}}} 

\newcommand{\mbb}{\bbbone_{2^{m}}}

\newcommand{\calA}{\mathcal{A}}
\newcommand{\calB}{\mathcal{B}}

\newcommand{\calD}{\mathcal{D}}

\newcommand{\calH}{\mathcal{H}}
\newcommand{\calI}{\mathcal{I}}
\newcommand{\calJ}{\mathcal{J}}
\newcommand{\calK}{\mathcal{K}}

\newcommand{\calS}{\mathcal{S}}

\newcommand{\calU}{\mathcal{U}}

\newcommand{\algA}{\calA}

\newcommand{\hs}{\calH}

\newcommand{\defeq}{\vcentcolon=} 

\DeclareMathOperator{\Aut}{Aut}

\DeclareMathOperator{\Inn}{Inn}

\DeclareMathOperator{\Tr}{Tr}	   

  \newcommand{\Ths}[1][]{\widetilde{\hs}}

\newcommand{\act}{\calS} 

\newcommand{\Sp}{\mathcal{S}} 
 \newcommand{\Dir}{D} 
\newcommand{\tw}{K}
\newcommand{\twr}{{\tw_{\scriptscriptstyle (1)}}}
\newcommand{\twpr}{{\tw_{\scriptscriptstyle (2)}}}
 
\newcommand{\rhr}{{\rho_{\scriptscriptstyle (1)}}}
\newcommand{\rhpr}{{\rho_{\scriptscriptstyle (2)}}}

\newcommand{\Tadj}{{\dagger_\tw}}

\newcommand{\Tadjpr}{{\dagger_\twpr}}
\newcommand{\KDir}{\Dir^{\tw }}
\newcommand{\TprDir}{\Dir^{\,\twpr }}

\newcommand{\rDir}{{\Dir_{\scriptscriptstyle (1)}}}

\newcommand{\prDir}{{\hat{\Dir}_{\scriptscriptstyle (2)}}}

\newcommand{\LDir}{\Dir^{L\, }}

\newcommand{\HDir}{\hat{\Dir}}
\newcommand{\HTDir}{ \hat{\Dir}^{\,\tw }}
\newcommand{\HTrDir}{ \hat{\Dir}^{\,\twr }}

\newcommand{\gpr}{{g_{\scriptscriptstyle PR}}} 
\newcommand{\gr}{{g_{\scriptscriptstyle R}}} 
\newcommand{\g}{g} 
\newcommand{\gmr}{\gamma_{\scriptscriptstyle R}} 
\newcommand{\gmpr}{\gamma_{\scriptscriptstyle PR}} 

\newcommand{\fnr}{\calJ_r^{\scriptscriptstyle R}}
\newcommand{\fnpr}{\calJ_r^{\scriptscriptstyle PR}}
\newcommand{\cl}{c} 
\newcommand{\clr}{{\cl_{\scriptscriptstyle R}}} 
\newcommand{\clpr}{{\cl_{\scriptscriptstyle PR}}}
\newcommand{\prr}{P_r^{\scriptscriptstyle R}} 
 \newcommand{\prpr}{P_r^{\scriptscriptstyle PR}} 
 
 \newcommand{\twphi}{{ \phi^{\,\tw}}}
  \newcommand{\twphir}{{ \phi^{\,\twr}}}
 \newcommand{\twphipr}{{ \phi^{\,\twpr}}}

\newcounter{mnotecount}[section]
\renewcommand{\themnotecount}{\thesection.\arabic{mnotecount}}
\newcommand{\mnote}[1]%
{\protect{\stepcounter{mnotecount}}${}^{\text{\footnotesize$\bullet$\themnotecount}}$%
\reversemarginpar%
\marginpar{\raggedleft\footnotesize$\bullet$\themnotecount: #1}}

\newlength{\mnotewidth}
\definecolor{blueamu}{RGB}{0, 101, 189}
\definecolor{cyanamu}{RGB}{61, 183, 228}
\newcommand{\dhorline}[3][0]{%
	\tikz[baseline=-2pt]{\path[decoration={markings, 
			mark=between positions 0 and 1 step 2*#3
			with {\node[color=blueamu, fill, circle, minimum width=#3, inner sep=0pt, anchor=south west] {};}},postaction={decorate}]  (0,#1) -- ++(#2,0);}}
\newcommand{\dvertline}[3][0]{%
	\tikz[baseline=2em]{\path[decoration={markings,
			mark=between positions 0 and 1 step 2*#2
			with {\node[color=black!50, fill, circle, minimum width=#2, inner sep=0pt, anchor=south west] {};}},postaction={decorate}] (0, #1) -- ++(0,#3);}}

\makeatletter\newcommand\HUGE{\@setfontsize\Huge{28}{0}}\makeatother		
\numberwithin{equation}{section}

\allowdisplaybreaks
\begin{document}
\renewcommand\figurename{Fig.}

{
	\makeatletter\def\@fnsymbol{\@arabic}\makeatother 
	\title{Signature change by a morphism\\ of spectral triples}


	\author{G. Nieuviarts\footnote{gaston.nieuviarts.wk@gmail.com}\\
		{\small DIMA, Universita di Genova}%
		\\
		\small{via Dodecaneso, 16146 Genova, Italia}\\[2ex]
	}
 
	\date{}
	
	\maketitle
	\vspace{-\baselineskip}
}  
  
\setcounter{tocdepth}{3}
 
\begin{abstract}
	We present a connection between twisted spectral triples and pseudo-Riemannian spectral triples, rooted in the fundamental interplay between twists and Krein products. A concept of morphism of spectral triples is introduced, transforming one spectral triple into its dual. In the case of even-dimensional manifolds, we demonstrate how this construction implements a local signature change via the parity operator induced by the twist. Consequently, the signature change transformation is governed solely by the unitary operator that implements the twist. This unitary is a central element of our approach, as it is directly linked to the Krein product, the twist, and the parity operator that implements the signature change.
\end{abstract}

\tableofcontents 
 
\newpage
 
\section{Introduction}
\label{sec introduction}

One of the cornerstones in the development of noncommutative geometry is the duality between topological spaces $\Man$ and the commutative algebras of smooth functions $C^\infty(\Man)$ pointed out by the Gelfand-Naimark theorem. The key idea that emerged was that we can describe “spaces” starting from such algebras. The algebra becomes the basic structure, and fundamental frameworks such as topology, Riemannian geometry, and differential calculus have been successfully translated within this new perspective. An important result known as the Connes reconstruction theorem (see \cite{connes2013spectral}) states that taking $C^\infty(\Man)$, $L^2(\Man, \Sp)$\footnote{The small $R$ that appears in symbols such as $\gmr$, $\nabla^{\scriptscriptstyle R,S}$ or
	$\gr$ indicates that the corresponding objects are associated with Riemannian structures.	Similarly, a small $PR$ will be used for pseudo-Riemannian structures.}, the Hilbert space of square-integrable sections of the spinor bundle over $\Man$, $\Dir=-i\gmr^a\nabla_a^{{\scriptscriptstyle R, S}}$, the Dirac operator with $\gmr^a$ the Riemannian gamma matrices and $\nabla_a^{{\scriptscriptstyle R, S}}$ the Levi-Civita spin connection, $J$ a given anti-unitary operator, and $\Gamma$ a grading, then spectral triples given by $(C^\infty(\Man), L^2(\Man, \Sp), \Dir, J, \Gamma)$ are in one-to-one correspondence with Riemannian spin manifolds, being a fundamental structure to understand the influence of Gravitation on spinor fields. Beyond the conceptual interest of this approach, an added value is that this reformulation allowed us to extend our geometrical conception by considering what this structure becomes when moving to general noncommutative algebras.

A significant achievement of noncommutative geometry has been the re-expression of Euclidean Gravitation and Yang-Mills theories within a unified geometric framework in which gauge and Higgs fields emerge as fluctuations of the Dirac operator. This is called the noncommutative standard model of particle physics; we refer to \cite{connes2006noncommutative, van2015noncommutative, lizzi2018noncommutative} for a detailed explanation of the model. However, this approach faces a fundamental problem: the model is inherently Euclidean. This is because Connes' reconstruction theorem works only for positive definite signatures. We therefore lose the characteristic causal structure of spaces with the Lorentz signature, thus missing out on physics. This limitation prevents an accurate description of the causal structure required in relativistic physics.

This weak point is at the origin of several attempts to find a Lorentzian version of the framework of spectral triples, such as \cite{barrett2007lorentzian, franco2014temporal, paschke2006equivariant, strohmaier2006noncommutative, bochniak2018finite, devastato2018lorentz, van2016krein, bizi2018space, d2016wick}. A common feature shared by most of these approaches is the replacement of the Hilbert space by a Krein space, turning the inner product $\langle \, .\, , \, .\, \rangle$ into an indefinite inner product $\langle\, . \, , \, . \, \rangle_\calJ\defeq \langle\, . \, , \, \calJ . \, \rangle$ with the so-called “fundamental symmetry” operator $\calJ=\calJ^\dagger$ so that $\calJ^2=\bbbone$ (see subsection \ref{SubsecKrein}). This comes from the fact that, as shown in \cite{baum1981spin}, a physical description of spinor fields on pseudo-Riemannian manifolds comes naturally with Krein spaces, so the use of such a structure in the quest for a Lorentzian approach within noncommutative geometry, within the so called pseudo-Riemannian spectral triples, as first implemented in \cite{strohmaier2006noncommutative}. A physical example of Krein space is given in quantum field theory, where the indefinite inner product is deduced from the operator $\calJ=\gmpr^1$ where $\gmpr^1$ is the first/temporal gamma matrix in space-time signature (+,-,-,-). 

Such an indefinite inner product recently appeared in \cite{devastato2018lorentz} and \cite{martinetti2024torsion} in the context of twisted spectral triples, first introduced by Connes and Moscovici in \cite{connes2006type} (see subsection \ref{SubsecIntrodTW}). In this framework, the derivative $[\Dir, a]$ is replaced by its twisted version $[\Dir, a]_\rho\defeq \Dir a-\rho(a)\Dir$ with $\rho$ being a regular automorphism called twist. The twist $\rho$ allows for the definition of the so-called $\rho$-inner product (see equation \eqref{EqRhoprod}). The Krein product then appears as a particular realization of the $\rho$-inner product, induced by twists of the form $\rho(\,.\, )=\calJ (\,.\, ) \calJ$, thus unexpectedly linking the Krein product to the twist $\rho$.

This article proposes a new perspective on the signature problem, based on the following two considerations:
\begin{enumerate}
\item If we consider the usual Lorentzian Dirac operator $\LDir=-i\gmpr^a\nabla_a^{{\scriptscriptstyle PR, S}}$, the physical Dirac action is given by $\act_D=\langle \psi , \LDir \psi \rangle_{\gmpr^1}$. A physically pertinent approach to the noncommutative standard model must produce such an invariant. Interestingly, the operator defined by $\Dir\defeq \gmpr^1 \LDir$ is self-adjoint, so that we also have $\act_D=\langle \psi , \Dir \psi \rangle$.
\item In the Riemannian case, Connes’ reconstruction theorem gives strong reasons to believe in the axioms of spectral triples. A key requirement in this framework is the self-adjointness of
the Dirac operator.
\end{enumerate}
The proposed approach is an attempt to reconcile these two considerations.

The primary contribution of this paper is to highlight the conceptual connection between twisted (Riemannian) spectral triple and pseudo-Riemannian spectral triples. This connection is a bijection implemented by a fundamental symmetry, see theorem \ref{Thmconnection}. 

The second, closely related contribution concerns the understanding of this connection in geometric terms. In the commutative case of even-dimensional spin manifolds, we show how the transition from a pseudo-Riemannian spectral triple with metric $\gpr$ to the connected twisted spectral triple is related to a change of metric signature i.e. the distance of the corresponding twisted spectral triple is related to the Riemannian metric $\gr$, connected to $\gpr$ through a space-like reflection $r$, deduced from $\tw$, see theorem \ref{DistanceTW}. 

A key result for physics concerns the equality of fermionic and spectral actions along the connection, see corollary \ref{PropDualFermAct}. This may have consequence for the noncommutative standard model approach, as discussed in the conclusion. This result is obtained in a compact setting and should therefore be interpreted at a local or algebraic level, rather than as a statement about globally hyperbolic Lorentzian spacetimes.
 
 In sections \ref{SecST} and \ref{SecTwsp} we present the defining axioms of Spectral Triples ($ST$) and Twisted Spectral Triples (TST). The relation between the Krein inner product and the $\rho$-inner product is presented in subsection \ref{SubsecKrein}, leading to the definition of the fundamental twist.

Then in section \ref{SecDual}, taking the twist $\rho(\,.\, )=\tw (\,.\, ) \tw^\dagger$ with $\tw$ being a fundamental symmetry, we define the so called $\tw$-Pseudo-Riemannian Spectral Triple ($\tw$-$PRST$), the $\tw$-Twisted Spectral Triple ($\tw$-$TST$) and then the $\tw$-Twisted Pseudo-Riemannian Spectral Triple ($\tw$-$TPRST$) in order to present the two following connections:

\begin{itemize}
	\item \textbf{connection 1:} between $\tw$-$TPRST$ and $ST$ 
	\item \textbf{connection 2:} between $\tw$-$PRST$ and $\tw$-$TST$.
\end{itemize}
The key concept of $\tw$-morphism of (generalized) spectral triples is introduced in subsection \ref{SubSecTwMorph} and permits to present the transformation between two spectral triples as a $\tw$-morphism between them.

In section \ref{SecMFLD}, the case of the spectral triples of even-dimensional compact manifolds is presented, showing in particular how the $\tw$-morphism realize a local signature change (only for an odd number of dimensions). The corresponding action of this reflection at the level of the Clifford representation is then implemented by the twist, which becomes a generalized parity operator in this framework. These results motivate in particular the definition of the so-called twisted Clifford algebra (see definition \ref{DefTwCliff}). The particular example of the 4-dimensional compact spin manifold is presented in subsection \ref{SubSec4dlor}, with a discussion on the Dirac action and the induced spectral triple, in view of future applications within the noncommutative standard model of particle physics.

\newpage
\section{Spectral triples}
\label{SecST}
We recall in this section some important properties of spectral triples, in particular concerning the way to generate the fluctuations of the Dirac operator together with the associated spectral invariants. Taking a $*$-algebra $\calA$ with representation $\pi$ on a Hilbert space $\calH$, we will omit the symbol of the representation by identifying the elements of $\calA$ with their representation. The adjoint in $\calB(\calH)$ is given by $\dagger$ so that $\pi(a^*)=\pi(a)^\dagger$.
 
\begin{definition}[Spectral triple]
 	A spectral triple $(\algA, \calH,D)$ is the data of an involutive unital algebra
 	$\algA$ represented by bounded operators on a Hilbert space $\calH$, and of a self-adjoint operator $D$ acting on $\calH$ such that the resolvent $(i + D^2)^{-1}$ is compact and that for any $a\in\algA$, $[D, a]$ is a bounded operator.
 \end{definition}
A real and even spectral triple $(\algA, \hs, \Dir, \Gamma, J)$ is defined by the introduction of two operators $\Gamma$ and $J$ in the spectral triple, such that $J$ is an anti-unitary operator and $\Gamma$ is a $\bbZ_2$-grading on $\hs$ satisfying $\Gamma^2 = \bbbone$ and $\Gamma^\dagger = \Gamma$ so that
\begin{equation}
J^2 = \epsilon,\qquad J D = \epsilon ' D J,\qquad J \Gamma = \epsilon'' \Gamma J,\qquad \Gamma \Dir + \Dir \Gamma = 0,\quad\text{and}\quad \Gamma a = a\Gamma
\end{equation}
$\forall a\in\algA$, with $\epsilon, \epsilon ', \epsilon ''=\pm 1$. 

In addition, the commutant property $[a, J b^\dagger J^{-1}] = 0$ and the first-order condition
\begin{equation}
	\label{FirstOrd}
[[D, a], J b^\dagger J^{-1}] = 0
\end{equation}
 must be satisfied $\forall a,b\in\calA$.\\
 
The derivative is given by $\delta(.)=[\Dir,\,.\,]$ so that the $\algA$-bimodule of Connes's differential one-forms is given by $\Omega_{\Dir}^{1}(\algA)\defeq\left\{\sum_{k} a_{k}[\Dir, b_{k}]: a_{k}, b_{k} \in \algA\right\}$. In this way, a one-form $A\in \Omega_{\Dir}^{1}(\algA)$ gives rise to a fluctuated Dirac operator defined by the following:
\begin{equation}
\Dir_A \defeq D + A + \epsilon' J A J^{-1}.
\end{equation}
 Taking $u \in \calU(\algA)$, we can define the adjoint unitary $U = u J u J^{-1}$ so that the inner fluctuation of $\Dir_A$ becomes equivalent to a gauge transformation of the one-form $A$:
 \begin{equation}
 (\Dir_A)^u=U\Dir_A U^\dagger= \Dir_{A^u}\qquad \text{with}\qquad A^u \defeq u A u^ \dagger + u \delta(u^ \dagger).
 \end{equation} 
Taking this, a way to obtain spectral invariants is given by the spectral and fermionic actions. The spectral action, introduced in \cite{chamseddine1997spectral} is given by:
\begin{align}
	\label{eqAction}
	\act[\Dir_A]
	&\defeq \Tr f(\Dir_A \Dir_A^\dagger / \Lambda^2)
\end{align}
with $f:\mathbb{R}^+\,\to\, \mathbb{R}^+$ a positive and even function making $f(\Dir_A \Dir_A^\dagger / \Lambda^2)$ be a trace-class operator, which decays at $\pm\infty$, $\Lambda\in\mathbb{R}^+$ being a cutoff parameter, see \cite{chamseddine2010noncommutative}. 

A fermionic action is defined by: 
\begin{equation}
\act_f( \Dir_A, \psi)\defeq \langle \psi , \Dir_A \psi \rangle\qquad\quad \text{with}\qquad\quad  \psi\in\calH.
\end{equation}

\newpage
\section{Twisted spectral triples}
\label{SecTwsp}

First introduced by A. Connes and H. Moscovici in \cite{connes2006type}, twisted spectral triples were then implemented in the context of the noncommutative standard model by A. Devastato and P. Martinetti in \cite{TwistSpontBreakDevastaMartine2017} then developed in \cite{TwistLandiMarti2016} and \cite{TwistGaugeLandiMarti2018} by G. Landi and P. Martinetti. This approach was first motivated by the hope that it can generate the missing scalar field necessary for the prediction of the Higgs mass, see \cite{TwistSpontBreakDevastaMartine2017}, \cite{filaci2021minimal}, and then that it can be connected with a pseudo-Riemannian geometry (see \cite{devastato2018lorentz}), and generate a Lorentz invariant fermionic action \cite{martinetti2022lorentzian}. Surprisingly, it was shown in \cite{martinetti2024torsion} that an orthogonal and geodesic preserving torsion is generated by the so called twisted fluctuation of the Dirac operator (defined later in equation \eqref{EqTwFluct}) on a 4-dimensional closed Riemannian spin manifold. In this context, a crucial role is played by the so-called $\rho$-unitary operators (see equation \eqref{EqRHoUnit}) which turn out to generate torsion and to implement Lorentz symmetry (see \cite{martinetti2024torsion}).

\subsection{Introduction to twisted spectral triples}
\label{SubsecIntrodTW}

\begin{definition}[Twisted spectral triple]
A twisted spectral triple is obtained by considering a triple $(\calA, \calH, \Dir)$ where $\calA$ is an involutive algebra acting as a bounded operator algebra on a Hilbert space $\calH$ with a self-adjoint operator $\Dir$ together with a regular automorphism $\rho\in \Aut(\calA)$ (i.e. satisfying $\rho(a^\dagger)=(\rho^{-1}(a))^\dagger\, \forall a\in\calA$) called twist. The twisted commutator
\begin{equation}
[\Dir,a]_\rho:= \Dir a-\rho(a)\Dir
\end{equation}
is used instead of the usual commutator $[\Dir,a]$ and we require only the twisted one to be bounded. A given twisted spectral triple is specified by the set $(\calA, \calH, \Dir, \rho)$.
\end{definition} 
We can then define the twisted derivation $\delta_\rho (a):= [\Dir,a]_\rho$ so that the twisted Leibniz rule $\delta_\rho (ab)=\rho(a)\delta_\rho (b)+\delta_\rho (a)b$ is verified $\forall a,b\in \calA$.

\begin{definition}[$\calB(\calH)$-regular automorphism]
	Given a Hilbert space $\calH$, an automorphism $\rho$ is said to be $\calB(\calH)$-regular if $\rho$ is an automorphism of $\calB(\calH)$ which satisfies the regularity condition on all $\calB(\calH)$.
\end{definition}

In the following, we require $\rho\in \Aut(\calA)$ to be $\calB(\calH)$-regular. Since \cite{devastato2018lorentz}, a novelty consists in the possibility of considering the so-called $\rho$-twisted product $\langle\, . \, , \, . \,   \rangle_\rho$, which by definition satisfies
\begin{equation}
	\label{EqRhoprod}
	\langle\psi   , O\psi^\prime   \rangle_\rho=\langle\rho(O)^\dagger\psi ,\psi^{\prime} \rangle_\rho\qquad \forall \psi, \psi^{\prime} \in \calH\qquad \text{and} \qquad  \forall O\in \calB(\calH).
\end{equation}
See Remark \ref{rqorigininner} for the justification of this choice.\\

Taking any $O\in\calB(\calH)$ and $\psi, \psi^\prime\in\calH$, the right and left adjoints are then defined by
\begin{align}
	&\langle O \psi , \psi^\prime \rangle_\rho=\langle  \psi , \rho^{-1}(O^\dagger) \psi^\prime \rangle_\rho\defeq \langle  \psi , O^{\dagger_R}\, \psi^\prime \rangle_\rho, \\
	&\langle  \psi , O \psi^\prime \rangle_\rho=\langle \rho(O)^\dagger  \psi , \psi^\prime \rangle_\rho\defeq \langle O^{\dagger_L}\, \psi , \psi^\prime \rangle_\rho.
\end{align}

The requirement that $\rho$ is $\calB(\calH)$-regular (resp. not $\calB(\calH)$-regular) then rephrases as the condition that the left and right adjoints satisfy $(\, .\,)^{\dagger_L}=(\, .\,)^{\dagger_R}$ (resp. $(\, .\,)^{\dagger_L}\neq(\, .\,)^{\dagger_R}$), inducing the uniqueness of the adjoint. The associated (unique) adjoint, called the $\rho$-adjoint, is then defined by
\begin{equation}
	(\, .\,)^+:= \rho(\, .\, )^\dagger.
\end{equation}

The proposition below establishes a fundamental link between the twist's regularity condition and the Hermitian structure of $\langle\, . \, , \, . \, \rangle_\rho$. 

\begin{proposition}
	If $\langle\, . \, , \, . \, \rangle_\rho$ is Hermitian, then $\rho$ is $\calB(\calH)$-regular. Conversely, if $\rho$ is not $\calB(\calH)$-regular, then $\langle\, . \, , \, . \, \rangle_\rho$ is not Hermitian.
\end{proposition}
\begin{proof}
	Taking any $O\in\calB(\calH)$ and $\psi, \phi\in\calH$, if $\langle\, . \, , \, . \, \rangle_\rho$ is Hermitian, then 
	\begin{equation}
		\langle  O^{\dagger_L}\psi , \phi \rangle_\rho=\langle  \psi , O\phi \rangle_\rho=\overline{\langle  O\phi , \psi \rangle}_\rho=\overline{\langle  \phi , O^{\dagger_R}\psi \rangle}_\rho=\langle O^{\dagger_R}\psi , \phi\rangle_\rho,
	\end{equation}
	which implies that $(\, .\,)^{\dagger_L}=(\, .\,)^{\dagger_R}$, and thus $\rho$ is $\calB(\calH)$-regular.
	
	If $\rho$ is not $\calB(\calH)$-regular, then $\langle  O^{\dagger_L}\psi , \phi \rangle_\rho\neq \langle O^{\dagger_R}\psi , \phi\rangle_\rho$. Now suppose that $\langle\, . \, , \, . \, \rangle_\rho$ is Hermitian. This means that 
	\begin{align}
		&\langle  O^{\dagger_L}\psi , \phi \rangle_\rho=\overline{\langle  O\phi , \psi \rangle}_\rho,\\
		&\langle O^{\dagger_R}\psi , \phi\rangle_\rho=\overline{\langle  \phi , O^{\dagger_R}\psi \rangle}_\rho=\overline{\langle  O\phi , \psi \rangle}_\rho
	\end{align}
	inducing the contradiction $\overline{\langle  O\phi , \psi \rangle}_\rho\neq \overline{\langle  O\phi , \psi \rangle}_\rho$, implying that $\langle\, . \, , \, . \, \rangle_\rho$ is not Hermitian. 
\end{proof}

As a consequence, the regularity condition is a necessary but not a sufficient condition for the inner product $\langle\, . \, , \, . \, \rangle_\rho$ to be Hermitian. In the case where $\rho\in\Inn(\calB(\calH))$, a sufficient condition will be given in Proposition \ref{PropTwistProd}. This underlines the deep connection between twists and Hermitian products.
 
 An operator $O\in \calB(\calH)$ is called “$\rho$-self-adjoint” if $O^+=O$. A $\rho$-unitary is an operator $U_\rho$ which is unitary with respect to the $\rho$-twisted product:
 \begin{equation}
 	\label{EqRHoUnit}
 	U_\rho U_\rho^+=U_\rho^+U_\rho=\bbbone.
 \end{equation}
 The set of all $\rho$-unitaries will be denoted by $\calU_\rho(\calB(\calH))$.
 
 \begin{remark}
 	$\calU_\rho(\calB(\calH))$ is the natural symmetry group for the $\rho$-twisted product, since the scalars $\langle \psi, \psi \rangle_\rho$ for $\psi\in\calH$ are preserved under the transformation $\psi\,\to\, U_\rho\psi$:
 	\begin{equation}
 		\langle U_\rho\psi, U_\rho\psi \rangle_\rho=\langle  \psi, U_\rho^+U_\rho\psi \rangle_\rho=\langle \psi, \psi \rangle_\rho. 
 	\end{equation}
 \end{remark}
 We will restrict our analysis to twists that are $*$-automorphisms, i.e., $\rho(O^\dagger)=\rho(O)^\dagger$. This condition, together with the regularity assumption, implies that $\rho^2=\bbbone$, which in turn ensures that $(\, . \,)^+$ is an involution. The axioms of twisted spectral triples are the same as those of standard spectral triples, except for the introduction of an additional characteristic number $\epsilon^{\prime\prime\prime}=\pm 1$ such that $\,\rho( J)=\epsilon^{\prime\prime\prime}J$, and the replacement of the first-order condition by the twisted first-order condition 
\begin{equation}
	\label{TwFirstorder}
[[\Dir, a]_\rho, b^\circ]_{\rho^\circ}=0 
\end{equation}
$\forall a\in\calA$ and $b^\circ \in\calA^\circ$ with $\calA^\circ$ the opposite algebra, with opposite twist $\rho^\circ$ acting on $\calA^\circ$ defined by $\rho^\circ(a^\circ):= (\rho^{-1}(a))^\circ$ (see \cite{TwistLandiMarti2016}). Twisted 1-forms are given by elements $A_\rho=\sum_ka_k[\Dir, b_k]_{\rho}$ with $a_k, b_k \in \calA$ so that twisted-fluctuations of $\Dir$ are given by
\begin{equation}
	\label{EqTwFluct}
	\Dir_{A_\rho}:= \Dir+A_\rho+\epsilon^\prime JA_\rho J^{-1}.
\end{equation}
As shown in \cite{TwistGaugeLandiMarti2018} and \cite{martinetti2024torsion}, there are two ways to generate such twisted-fluctuations:
\begin{enumerate}
	\item Type 1: $\Dir_{A_\rho}=U\Dir U^+$ with $U=uJuJ^{-1}$ and $u$ a unitary operator in $\calA$.
	\item Type 2: $\Dir_{A_\rho}=U_\rho\Dir U_\rho^\dagger$ with $U_\rho=u_\rho Ju_\rho J^{-1}$ and $u_\rho$ a $\rho$-unitary operators in $\calA$.
\end{enumerate}
We refer to \cite{TwistGaugeLandiMarti2018} for the understanding of the connection with Morita equivalence.

\begin{remark}
	\label{rqorigininner}
	The fact that $U\Dir U^+$ preserves $\rho$-self-adjointness but not the usual self-adjointness was the original motivation for introducing the $\rho$-inner product in \cite{devastato2018lorentz}. The discovery of an alternative method to generate twisted fluctuations, where $U_\rho$ is $\rho$-unitary (preserving self-adjointness and producing skew torsion terms), through the use of a $\rho$-unitary operator in \cite{martinetti2024torsion}, shows that both inner products can be used in this framework.
\end{remark}

Following \cite{martinetti2022lorentzian}, the spectral action remains the same as for standard spectral triples, except that $\Dir_{A_\rho}$ is used instead of $\Dir_{A}$. The fermionic action then becomes the $\rho$-twisted fermionic action, defined by:
\begin{equation}
	\act_f^\rho(\Dir_{A_\rho}, \psi)\defeq \langle \psi , \Dir_{A_\rho} \psi \rangle_\rho\qquad\quad \text{with}\qquad\quad  \psi\in\calH.
\end{equation} 

Starting from a graded spectral triple $(\algA, \calH, \Dir, \Gamma)$, we ask how to twist this spectral triple in a minimal way, i.e., without modifying the Hilbert space or the Dirac operator, but only the algebra. A motivation for twisting the spectral triple in this way comes from the fact that the fermionic content of the Standard Model of particle physics is encoded within $\Dir$ and $\calH$, making any modification to these fundamental elements more risky than doubling the algebra as done below, see \cite{TwistSpontBreakDevastaMartine2017, TwistLandiMarti2016} for more arguments.
 
As shown in \cite{TwistLandiMarti2016} a way to  “minimally twist”\footnote{Without changing the Hilbert space and the Dirac operator} $(\algA, \calH, \Dir, \Gamma)$ is by the use of the decomposition coming with $\Gamma$. Using the projections $p_\pm:=  (\bbbone \pm \Gamma)/2$ so that $\calH=p_+\calH\oplus p_-\calH:=\calH_+\oplus \calH_-$, all elements $a=(a_1,a_2)\in \algA\otimes \mathbb{C}^2$ are represented accordingly:
\begin{equation}
	\label{EqRhogen}
	\pi(a_1,a_2)\defeq p_+\pi_0(a_1)+p_-\pi_0(a_2)=	\begin{pmatrix}
		\pi_+(a_1) & 0  \\
		0 & \pi_-(a_2)  \\
	\end{pmatrix}
\end{equation}
with $\pi_0$ the representation used for $\algA$ and $\pi_\pm(\, .\, )\defeq p_\pm \pi_0(\, .\, )_{|\calH_\pm}$ the restrictions to $\calH_\pm$. This will be called the chiral representation.
\begin{definition}[twist by grading]
	Taking the graded spectral triple $(\algA, \calH, \Dir, \Gamma)$, its corresponding twist by grading is given by the twisted spectral triple $(\algA\otimes \mathbb{C}^2 , \calH, \Dir, \Gamma, \rho )$, using the chiral representation. The twist acts on elements $a=(a_1,a_2)\in \algA\otimes \mathbb{C}^2$ as
	\begin{equation}
		\label{EqactionRho}
		\rho(a_1,a_2)= (a_2,a_1).
	\end{equation}  
\end{definition} 
Now consider the spectral triple for a $2m$-dimensional manifold $\Man$ with $m\in\mathbb{N}$. The Hilbert space $\calH$ can be split according to the two “chiral” eigenspaces of $\Gamma$: $\calH=\calH_+\oplus \calH_-$. In this way, any element $a\in C^\infty(\Man)$ acts as $a=(f,f)\equiv (f\bbbone_{2^{m-1}}, f\bbbone_{2^{m-1}})$\footnote{This comes from the constraint that the grading must commute with the algebra.} on elements $\psi=(\psi_+, \psi_-)$ with $\psi_\pm\in \calH_\pm$ i.e. in a symmetrical way, according to $\calH_\pm$. The doubled algebra is $\calA\defeq C^\infty(\Man)\otimes \mathbb{C}^2$ whose elements are given by $a=(f, g)$ where $\calA^\prime\defeq C^\infty(\Man)$ corresponds to its restriction to elements so that $f=g$. We have $\forall a\in\calA$ and $b\in\{1, \dots, 2m\}$ that the gamma matrices twist commutes with the algebra, i.e. $[\gmr^b, a]_\rho=0$ where the twist is given by the action $\rho(a)\equiv\rho(f,g)=(g,f)$. Then, taking $\psi\in\calH$, we have that $\Dir(a\psi)=\rho(a)\Dir(\psi)+\Dir(a)\psi$  so that the commutator with the Dirac operator is given by 
\begin{equation}
	[\Dir, a]=\Dir(a)+(\rho(a)-a)\Dir.
\end{equation}
This commutator is now unbounded since the last term $(\rho(a)-a)\Dir$ isn't equal to $0$ on $\calA$, and does not provide the desired differential $\Dir (a)=-ic(d(a))$ where $c$ is the Clifford action. We are left with two possibilities to obtain a bounded commutator. Either we restrict to $\calA^\prime$ i.e. the sub-algebra of $\calA$ containing elements $a$ satisfying $\rho(a)=a$ or we can replace the commutator by the twisted commutator
\begin{align}
	[\Dir, a]_\rho=\Dir a-\rho(a)\Dir= \Dir (a)
\end{align}
which is indeed bounded and gives the differential. This is the only way to obtain a bounded commutator in this case (see \cite{TwistLandiMarti2016}). 
\begin{remark}
It is interesting to notice that the doublet $(\calA, [\Dir, \, .\,]_\rho)$ is a natural extension of the one of the previous case $(\calA^\prime, [\Dir, \, .\,])$ since $[\Dir, a]_\rho=[\Dir, a]$ for $a\in\calA^\prime$. Twisted spectral triples are then natural extensions of spectral triples, which correspond to the restriction of the twisted ones to the sub-algebra $\calA^\prime$ one which the twist acts trivially, i.e. $\rho(a)=a$ for any $a\in \calA^\prime$. It is therefore by considering algebras that distinguish the $\calH_\pm$ chiral spaces that twisted spectral triples came naturally.
\end{remark}

\subsection{Twisted product when $\rho$ is an inner automorphism} 

In the case $\rho\in\Inn(\calB(\calH))$, there is a unitary operator $\tw$ acting on $\calH$ such that $\rho(O)=\tw O \tw^\dagger$ and the $\rho$-twisted product can then be defined from the hermitian inner product $\langle\, . \, , \, . \, \rangle$ and $\tw$ such that $\langle\, . \, , \, . \, \rangle_\rho := \langle\, . \, , \tw . \, \rangle$. In this setting, the twist defines a $*$-automorphism of $\calA$. In the following, any twist $\rho$ will be an element of $\Inn(\calB(\calH))$. Since different unitary $\tw$ can generate the twist, we will take more adapted notations to specify the algebraic structures. The twisted spectral triple $(\calA, \calH, \Dir, \rho)$ will then be denoted as $(\calA, \calH, \Dir, \tw)$, turning the $\rho$-inner product notation $\langle\, . \, , \, . \, \rangle_\rho$ into the $\tw$-inner product one. 
\begin{equation}
	\label{Eqtwprod}
	\langle\, . \, , \, . \, \rangle_\tw := \langle\, . \, , \tw . \, \rangle
\end{equation}
so that the associated adjoint is called $\tw$-adjoint and is defined by
\begin{equation}
	\label{EqDefAdjInn}
(\, .\,)^\Tadj:= \tw(\, .\,)^\dagger\tw^\dagger
\end{equation}
and $\tw$-unitary operators $U_\tw U_\tw^\Tadj=U_\tw^\Tadj U_\tw=\bbbone$ with corresponding set called $\calU_\tw(\calB(\calH))$.

\begin{proposition}
	\label{PropNormtwAdj}
	Taking the norm $\|(\, .\,)\|$ induced by $\langle\,   \, ,   \, \rangle$, for any operator $O\in\calB(\calH)$ we have that $\|O\|=\|O^\Tadj\|$.
\end{proposition}
\begin{proof}
	We have that $\|O^\Tadj\|=\|\tw O^\dagger\tw^\dagger\|=\| O^\dagger\tw^\dagger\|$ since $\tw$ is a unitary operator. Then using the relation $\|(\, .\, )\|=\|(\, .\, )^\dagger\|$ we get that $\|O^\Tadj\|=\|\tw O\|=\| O\|$.
\end{proof}

\begin{lemma}
	\label{PropRegImpFund}
Taking $\rho\in\Inn(\calB(\calH))$ implemented by the unitary $\tw$, If $\calA$ is a simple finite-dimensional algebra or if $\rho$ is $\calB(\calH)$-regular, then it follows that
\begin{equation}
	\label{EqFormTw}
\tw=\exp(i\theta)\tw^\dagger
\end{equation}
for $\theta$ a real number. 
\end{lemma}
 \begin{proof}
Lets take an operator $O\in \calB(\calH)$, we have that $\rho(O^\dagger)=\tw O^\dagger \tw^\dagger$ and $\rho^{-1}(O)=\tw^\dagger O \tw$. The regularity condition $\rho(O^\dagger)=(\rho^{-1}(O))^\dagger$ implies $\tw O^\dagger \tw^\dagger= \tw^\dagger O^\dagger \tw$. The last equality is equivalent to $O^\dagger=(\tw^\dagger)^2 O^\dagger(\tw)^2$ meaning that $(\tw)^2$ is in the commutant of $\calB(\calH)$ and then is of the form $\lambda \bbbone$ for $\lambda\in\mathbb{C}$. We then have $(\tw)^2=\lambda \bbbone$ and $(\tw^\dagger)^2=\bar\lambda \bbbone$ so that $\tw=\sqrt{\lambda}\tilde{\tw}$ and $\tw^\dagger=\sqrt{\bar\lambda}\tilde{\tw}^\dagger$ with $(\tilde{\tw})^2=(\tilde{\tw}^\dagger)^2=\bbbone$ and $\lambda\bar\lambda=1$. Moreover, we have that $\tw\tw^\dagger =\tilde{\tw}\tilde{\tw}^\dagger=\bbbone$ implying that $\tilde{\tw}=\tilde{\tw}^\dagger$ so that taking $\lambda=\exp(i\theta)$, we are left with $\sqrt{\lambda^{-1}}\tw=\sqrt{\bar{\lambda}^{-1}}\tw^\dagger$ which gives $\tw=\exp(i\theta)\tw^\dagger$. The conclusion is the same if $O$ is replaced by an element of a simple finite dimensional algebra $\calA$ since its commutant is also proportional to identity.
 \end{proof}
\begin{corollary}
If $\calA$ is a simple finite dimensional algebra, and $\rho$ is regular on $\calA$, then $\rho$ is $\calB(\calH)$-regular.
\end{corollary}
\begin{proof}
It was shown in lemma \ref{PropRegImpFund} that $\tw=\exp(i\theta)\tw^\dagger$ in this case. As a consequence $\rho(\, .\,)=\tw(\, .\,) \tw^\dagger=\tw^\dagger(\, .\,) \tw=\rho^{-1}(\, .\,)$ and then $\rho((\, .\,)^\dagger)=\rho^{-1}((\, .\,)^\dagger)=\tw^\dagger (\, .\,)^\dagger\tw =(\tw^\dagger (\, .\,)\tw)^\dagger=\rho^{-1}(\, .\,)^\dagger$ so that $\rho$ is also $\calB(\calH)$-regular.
\end{proof}

\begin{proposition}
	\label{PropInnTw}
	In the case of the twist by grading, we have that $\rho\in\Inn(\calB(\calH))$ so that $\rho(\,.\,)=\tw(\,.\,)\tw^\dagger$ with $\tw$ a unitary operator in $\calB(\calH)$ such that
	\begin{equation}
		\label{EqForTwGrad}
		\tw =	\begin{pmatrix}
			0 & \tw_2 \\
			\tw_3 & 0 \\
		\end{pmatrix}
	\end{equation}
	with $\tw_2$ and $\tw_3$ being unitaries so that $\forall a=(a_1,a_2)\in \algA\otimes \mathbb{C}^2$ we have $[\tw_2, a_2]=[\tw_3, a_1]=0$.
\end{proposition}
\begin{proof}
	Lets suppose that $\rho(a)=\tw a \tw^\dagger$ with $a=(a_1,a_2)$ and
	\begin{equation}
		\tw =	\begin{pmatrix}
			\tw_1 & \tw_2 \\
			\tw_3 & \tw_4  \\
		\end{pmatrix}
	\end{equation}
	in the chiral representation. Then the requirement that  $\rho(a_1,a_2)= (a_2,a_1)$ implies that $\tw_1=\tw_4=0$ so that $\tw_2 a_2 \tw_2^\dagger =a_2$ and $\tw_3 a_1 \tw_3^\dagger =a_1$.  Using that $\tw$ is unitary, we then get that $\tw_2$ and $\tw_3$ are also unitaries. The fact that $\tw_2 a_2 \tw_2^\dagger =a_2$ and $\tw_3 a_1 \tw_3^\dagger =a_1$ implies that $[\tw_2, a_2]=[\tw_3, a_1]=0$ for all $a\in \algA\otimes \mathbb{C}^2$, hence the result of equation \eqref{EqForTwGrad}.
\end{proof}
A trivial and always existing solution is given by taking $\tw_2=\tw_3=\bbbone$. 

%
 
\begin{definition}[Trivial twist]
	Taking an algebra $\calA$, a twist $\rho$ is said to be trivial on $\calA$ if $\rho(a)=a$ for any $a\in\calA$.
\end{definition}
An example of a trivial twist is given by the twist by the grading on $\calA^\prime= C^\infty(\Man)$ which contains elements $a=(f,f)$. Another example is the twist $\rho(\, .\,)=\tw(\, .\,)\tw^\dagger$ where $\tw=\bbbone$. 

Taking $\rho$ to be an inner $\calB(\calH)$-regular automorphism implemented by the unitary $\tw$, the following proposition highlights how the restriction of the corresponding twisted inner products to the Hermitian product leads to the fact that $\tw=\tw^\dagger$, inducing that $\langle\, . \, , \, . \, \rangle_\tw$ is an indefinite inner product.

\begin{proposition}
	\label{PropTwistProd}
	Any non-trivial $\calB(\calH)$-regular twist $\rho$ implemented by a unitary operator $\tw$ induces a $\tw$-inner product $\langle\, . \, , \, . \, \rangle_\tw = \langle\, . \, , \tw . \, \rangle$ which is Hermitian and indefinite if and only if $\tw=\tw^\dagger$. In this case, $\langle \, .\, , \, .\, \rangle_{\tw}$ is Hermitian and indefinite.
\end{proposition}

\begin{proof}
	From lemma \ref{PropRegImpFund}, the $\calB(\calH)$-regularity condition implies that $\tw$ satisfies equation \eqref{EqFormTw}. Since the inner product $\langle\, . \, , \, . \, \rangle$ is Hermitian, taking $\psi,\psi^\prime$ in $\calH$ we have that
	\begin{equation}
		\langle \psi , \psi^\prime \rangle_\tw =\langle \psi , \tw \psi^\prime \rangle=\overline{\langle \tw\psi^\prime ,  \psi \rangle}=\overline{\langle \psi^\prime ,  \tw^\dagger\psi \rangle}=\exp(i\theta)\overline{\langle \psi^\prime ,  \psi \rangle}_\tw.
	\end{equation}	
	This implies that the inner product is Hermitian if and only if $\theta=2n\pi$ for some integer $n$. This induces that $\tw=\tw^\dagger$ and that $\langle\, . \, , \, . \, \rangle_\tw$ is also indefinite.
	
	A non-trivial twist, i.e. $\rho(a)\neq a$ for any $a\in\calA$, is always realized by a unitary $\tw\neq \pm \bbbone$. In order to determine when the induced $\tw$-inner product $\langle\, . \, , \, . \, \rangle_\tw$ is indefinite, we compute the eigenvalues of $\tw$. Equation \eqref{EqFormTw} implies that $\tw^2=\exp(i\theta)$. We can then construct the projection operators $P_\pm\defeq (\bbbone\pm \exp(-i\theta/2)\tw)/2$ so that $P_++P_-=\bbbone$ with $P_+P_-=0$ and $\tw P_\pm=(\tw\pm \exp(i\theta/2)\bbbone)/2=\pm\exp(i\theta/2) P_\pm$. Taking a vector $\psi\in V$, we define $\psi_\pm\defeq P_\pm \psi$ so that $\psi=\psi_++\psi_-$ and consequently
	\begin{equation}
		\tw \psi_\pm=\tw P_\pm \psi=\pm\exp(i\theta/2) P_\pm \psi=\pm\exp(i\theta/2) \psi_\pm.
	\end{equation} 
	Then $\psi_\pm$ are eigenvectors of $\tw$ with the eigenvalues $\pm\exp(i\theta/2)$, and we obtain
	\begin{equation}
		\label{EqPhaseRelaInners}
		\langle\psi_\pm , \psi_\pm \rangle_\tw=\pm\exp(i\theta/2)\langle\psi_\pm , \psi_\pm \rangle,
	\end{equation}
	where $\langle\psi_\pm , \psi_\pm \rangle$ is always positive and real.
	
Thus, $\langle\, . \, , \, . \, \rangle_\tw$ is positive definite or indefinite precisely when $\theta=2n\pi$ for some integer $n$, ensuring that $\tw=\tw^\dagger$ and making the inner product Hermitian. In this case, in order to show that the induced $\tw$-inner product $\langle\, . \, , \, . \, \rangle_\tw$ is always indefinite, we prove that $\tw$ always admits $1$ and $-1$ as eigenvalues when the twist is non-trivial. The projection operators become $P_\pm\defeq (\bbbone\pm \tw)/2$, so that $\tw=P_+-P_-$ and $\tw P_\pm=\pm P_\pm$. Taking a vector $\psi\in V$ with $\psi_\pm\defeq P_\pm \psi$, we obtain $\tw \psi_\pm=\tw P_\pm \psi=\pm P_\pm \psi=\pm \psi_\pm$. Since the twist is non-trivial, $\tw\neq \pm \bbbone$, which implies that $P_+\neq0$ and $P_-\neq0$, ensuring the existence of both eigenspaces and the eigenvalues $\pm1$. Since $\langle\, . \, , \, . \, \rangle$ is positive definite, it follows that
	\begin{equation}
		\langle \psi_\pm , \psi_\pm \rangle_\tw = \pm\langle \psi_\pm, \psi_\pm \rangle,
	\end{equation}
	which implies that $\langle\, . \, , \, . \, \rangle_\tw$ is indefinite. Finally, since $\tw^2=\bbbone$, we recover the positive definite inner product via the relation $\langle\, . \, , \, . \, \rangle=\langle\, . \, , \tw \, . \, \rangle_\tw$, ensuring that $\tw=\tw^\dagger$.
\end{proof}
Hermitian (indefinite) inner products then appear as a special case of twisted products associated with non-trivial $\calB(\calH)$-regular twists, where the phase shift in equation \eqref{EqFormTw} vanishes, i.e., $\tw=\tw^\dagger$.
 \begin{remark}
 	The only way to obtain a positive definite product $\langle\, . \, , \, . \, \rangle_\tw$ is when $\tw=\bbbone$ i.e. in the case $\rho$ is trivial on all $\calB(\calH)$. This case correspond to usual spectral triples.
 \end{remark}
 \begin{definition}[Fundamental twist]
 	A fundamental twist is a twist $\rho\in\Inn(\calB(\calH))$, implemented by the unitary $\tw$ so that $\tw=\tw^\dagger$. 
 \end{definition}
 
\newpage
\section{Pseudo-Riemannian spectral triples} 
 \label{SubsecKrein}
 The main purpose of this subsection is to present the connection between the so-called Krein inner product and the $\tw$-inner product, leading to the definition of the pseudo-Riemannian spectral triple and it's twisted version.
\begin{definition}[Krein space]
Lets consider a vector space $V$ together with an indefinite (nondegenerate) inner product $\langle \, .\, , \, .\, \rangle_{\calJ}$ such that $V=V^+\oplus V^-$ with $\langle \, .\, , \, .\, \rangle_{\calJ}$ being positive definite on $V^+$ and negative definite on $V^-$. If  $V^\pm$ are complete for the induced norm, then $\calK\defeq (V, \langle \, .\, , \, .\, \rangle_{\calJ})$ is called Krein space, and $\langle \, .\, , \, .\, \rangle_{\calJ}$ is the corresponding Krein product. 
\end{definition}
\begin{definition}[Fundamental symmetry of a Krein space]
Taking the Krein space $\calK=(V, \langle \, .\, , \, .\, \rangle_{\calJ})$ with orthogonal projections $P_{\pm}$ on $V_{\pm}$ i.e. $P_{\pm}V=V_{\pm}$, the associated fundamental symmetry is the operator $\calJ\defeq P_+-P_-$. This operator satisfy the relations $\calJ^\dagger=\calJ$ and $\calJ^2=\bbbone$ and induces a positive definite inner product $\langle \, .\, , \, .\, \rangle$ by the relation
\begin{equation}
\langle \, .\, , \, .\, \rangle\defeq \langle \, .\, , \,\calJ .\, \rangle_{\calJ}.
\end{equation}
\end{definition}

Conversely, starting from $\calJ$ with $\calJ^\dagger=\calJ$ and $\calJ^2=\bbbone$ so that $\calJ$ have two eigenspaces $V^\pm$ corresponding to eigenvalues $\pm 1$ whose direct sum is $V$, we can define the indefinite inner product $\langle \, .\, , \, .\, \rangle_{\calJ}	\defeq \langle \, .\, , \calJ\, .\, \rangle$ for which $\calJ$ is a fundamental symmetry.

 As a consequence of proposition \ref{PropTwistProd}, any non-trivial fundamental twist $\rho(\, .\,)=\tw(\, .\,)\tw$ induces a Krein product $\langle\, . \, , \, . \, \rangle_\tw$ with fundamental symmetry $\tw$.
 
 \begin{lemma}
If the twist by grading is fundamental, then the unitary $\tw$ is given by 
	\begin{equation}
	\label{EqCompJ}
	\tw =	\begin{pmatrix}
		0 & \sigma \\
		\sigma^\dagger & 0 \\
	\end{pmatrix}
\end{equation}
with $\sigma$ being a unitary matrix so that $\forall a=(a_1,a_2)\in \algA\otimes \mathbb{C}^2$ we have $[\sigma, a_2]=[\sigma^\dagger, a_1]=0$.
  \end{lemma}
 \begin{proof}
 	Taking the result of equation \eqref{EqForTwGrad} and imposing the twist to be fundamental, we get that $\tw^2=\bbbone$ implies $\tw_2\tw_3=\tw_3\tw_2=\bbbone$. Multiplying on the left by $\tw_2^\dagger$ we get that $\tw_3=\tw_2^\dagger$. Changing the notation $\tw_2$ into $\sigma$ we obtain the result of equation \eqref{EqCompJ}. 
 \end{proof}

Following \cite{strohmaier2006noncommutative}, we now define two notions of spectral triples, based on Krein spaces.
\begin{definition}[Pseudo-Riemannian spectral triple]
 A pseudo-Riemannian spectral triple $(\calA,\calK, \tilde{\Dir})$ is a spectral triple for which the Hilbert space is turned into a Krein space $\calK$ and the Dirac operator into a Krein selfadjoint operator $\tilde{\Dir}$ with the requirement that $[\tilde{\Dir} ,a]$ is bounded for any $a\in\calA$. 
\end{definition}

\begin{definition}[Twisted pseudo-Riemannian spectral triple]
A twisted pseudo-Riemannian spectral triple $(\calA,\calK, \tilde{\Dir}, \tw )$ is a pseudo-Riemannian spectral triple for which only $[\tilde{\Dir}, a ]_\rho$ is required to be bounded for any $a\in\calA$. 
\end{definition}

\newpage

\section{Connection for even real spectral triples}
\label{SecDual}
The purpose of this section is to present the four kinds of generalized spectral triples and to show the bijective correspondence between their structures, i.e. their Dirac operators, derivations, one-forms, axioms, and fluctuations. The concept of $\tw$-morphism is presented in subsection \ref{SubSecTwMorph}. It will implement the transition from one generalized spectral triple to its dual one, thanks to the unitary $\tw$.
 
\subsection{Presentation of the connection}
 \label{SubSecPresDual}

In the following, we will only consider non-trivial fundamental twists $\rho\in\Inn(\calB(\calH))$ for $\calA$ an involutive algebra acting as bounded operator on a Hilbert space $\calH$. We then have $\rho(\, .\,)=\tw (\, .\,) \tw$ with $\tw=\tw^\dagger$ and $\tw^2=\bbbone$. The induced indefinite inner product is given by $\langle\, . \, , \, . \, \rangle_\tw := \langle\, . \, , \tw . \, \rangle$ for which $\tw$ is a fundamental symmetry. Starting from a Hilbert space $\calH\defeq (V, \langle \, .\, , \, .\, \rangle)$ where $V$ is a vector space, we call $\calK_\tw\defeq  (V, \langle \, .\, , \, .\, \rangle_\tw)$ the corresponding Krein space. The norm is the one induced by the inner product $\langle \, .\, , \, .\, \rangle$. In the same way, $\Dir$ will design a selfadjoint Dirac operator and $\KDir$ a $\tw$-selfadjoint Dirac operator.

Taking a fundamental symmetry $\tw$, the relation $O=\tw P$ between any two operators $O$ and $P$ in $\calB(\calH)$ defines a bijective correspondence between these operators as the transformation $P\, \to\, O=\tw P$ is a bijection.
 
\begin{proposition}
	\label{DualDir}
	If $\tw$ is a fundamental symmetry, then $\tw$-self-adjoint operators $O^{\, \tw}$ and self-adjoint operators $O$ are in bijective correspondence through the relation
	\begin{equation}
	\label{EqDualD}
	O^{\, \tw}\defeq\tw O.
	\end{equation}
\end{proposition}
\begin{proof}
	If $O$ is self-adjoint then $O^{\, \tw}= \tw O $ is $\tw$-self-adjoint:
	\begin{equation}
		(O^{\, \tw})^\Tadj=(\tw O^{\, \tw} \tw^\dagger)^\dagger=(\tw \tw O \tw^\dagger)^\dagger=\tw O=O^{\, \tw}.
	\end{equation}
Since $\tw$ is a fundamental symmetry, equation \eqref{EqDualD} is equivalent to $O= \tw O^{\, \tw}$. Then, the fact that $O^{\, \tw}$ is $\tw$-self-adjoint implies that $O$ is self-adjoint since $O^\dagger= (\tw O^{\, \tw})^\dagger =(O^{\, \tw})^\dagger \tw^\dagger=\tw O^{\, \tw}=O$. The relation between $O$ and $O^{\, \tw}$ is bijective as $\tw$ is unitary. 
\end{proof}
As a consequence of proposition \ref{DualDir}, $\tw$-self-adjoint Dirac operators $\KDir$ and self-adjoint Dirac operators $\Dir$ are in bijective correspondence through the relation $\KDir=\tw \Dir$. This defines the first element of the connection, for the bijective connection between the Dirac operators of spectral triples and pseudo-Riemannian spectral triples.

We can now define the following four kinds of generalized spectral triples. 
\begin{definition}[generalized spectral triples]
	\label{Def4kindsST}
Taking the involutive algebra $\calA$, the unitary operator $\tw=\tw^\dagger$, the Dirac operators $\KDir$ and $\Dir$ related by $\KDir\defeq \tw\Dir $ with the real operator $J$ and the grading $\Gamma$, we define four specific kind of generalized spectral triple as
\begin{itemize}
	\item $ST$: Spectral Triple $(\calA,\calH, \Dir, J, \Gamma )$,
	\item $\tw$-$PRST$: $\tw$-Pseudo-Riemannian Spectral Triple $(\calA,\calK_\tw, \KDir, J, \Gamma )$,
	\item $\tw$-$TST$: $\tw$-Twisted Spectral Triple $(\calA,\calH, \Dir, J, \Gamma, \tw )$,
	\item $\tw$-$TPRST$: $\tw$-Twisted Pseudo-Riemannian Spectral Triple $(\calA,\calK_\tw, \KDir, J, \Gamma, \tw )$.
\end{itemize}
\end{definition}
where the operator $\tw$ defines the three last kinds of spectral triples, starting from the usual spectral triple $(\calA,\calH, \Dir, J, \Gamma )$. The key point here is the connection of the Krein inner product of the pseudo-Riemannian spectral triples with the twist $\rho$. Note that the choice was done to use twisted (Krein) products only in the pseudo-Riemannian case.

\begin{remark}
We have that the $\tw$-$PRST$, the $\tw$-$TST$ and the $\tw$-$TPRST$ give back the $ST$ if we take $\tw=\bbbone$, presenting them as a deformation (with parameter $\tw$) of the usual $ST$. 
\end{remark}

In the following, we denote by $A$ one-forms in the context of the $ST$, by $A^{\,\tw}$ one-forms in the context of $\tw$-$PRST$, by $A_\rho$ twisted one-forms in the context of $\tw$-$TST$ and by $A_\rho^{\,\tw}$ twisted one-forms in the context of $\tw$-$TPRST$.

We can now determine the form fluctuations of the Dirac operator takes in the context where the Dirac operator is taken to be $\KDir=\tw \Dir$.
\begin{proposition}
	Taking the operator $\KDir$, twisted and usual fluctuations are given by
	\begin{align*}
		&\KDir_{A_\rho^{\,\tw}}=\KDir+ A_\rho^{\,\tw} +\epsilon^{\prime}\epsilon^{\prime\prime\prime} J A_\rho^{\,\tw} J^{-1}\qquad\quad&\text{ for $\tw$-$TPRST$}\\
		&\KDir_{A^{\,\tw}}=\KDir+ A^{\,\tw} +\epsilon^{\prime}\epsilon^{\prime\prime\prime} J A^{\,\tw} J^{-1}\qquad\quad&\text{for $\tw$-$PRST$}
	\end{align*}
\end{proposition}
\begin{proof}
	Taking the unitary operator $U=uJuJ^{-1}=u\hat{u}=\hat{u}u$, with $u$ and $\hat{u}\defeq JuJ^{-1}$ unitary operators and noting that $J\Dir=\epsilon^\prime \Dir J$ implies $J\KDir= \epsilon^\prime\epsilon^{\prime\prime\prime}\KDir J$, then we have for twisted fluctuations:
	\begin{align*}
		U\KDir U^\Tadj&=u\hat{u} \KDir \hat{u}^\Tadj u^\Tadj =\hat{u}(u\KDir u^\Tadj)\hat{u}^\Tadj = \hat{u}(\KDir+ u [\KDir,u^\Tadj]_\rho)\hat{u}^\Tadj\\
		&= \hat{u}\KDir \hat{u}^\Tadj + \hat{u}u\hat{u}^\dagger [\KDir,u^\Tadj]_\rho \qquad \text{(twisted first-order condition.)}\\
		&= \hat{u}\KDir \hat{u}^\Tadj + \hat{u}\hat{u}^\dagger u [\KDir,u^\Tadj]_\rho \qquad \text{(zero-order condition.)}\\
		&= JuJ^{-1}\KDir J u^\Tadj J^{-1} +  u [\KDir,u^\Tadj]_\rho\\
		&= \epsilon^{\prime}\epsilon^{\prime\prime\prime} JuJ^{-1}J\KDir  u^\Tadj J^{-1} +  u [\KDir,u^\Tadj]_\rho\qquad \text{(via $J\KDir= \epsilon^\prime\epsilon^{\prime\prime\prime}\KDir J$)}\\
		&= \epsilon^{\prime}\epsilon^{\prime\prime\prime}J(\KDir+u[\KDir,u^\Tadj]_\rho)J^{-1}+  u [\KDir,u^\Tadj]_\rho\\
		&=\KDir+ u [\KDir,u^\Tadj]_\rho +\epsilon^{\prime}\epsilon^{\prime\prime\prime} J u[\KDir,u^\Tadj]_\rho J^{-1}
	\end{align*}
	so that twisted fluctuations of $\KDir$ take the form $\KDir_{A_\rho^{\,\tw}}=\KDir+ A_\rho^{\,\tw} +\epsilon^{\prime}\epsilon^{\prime\prime\prime} J A_\rho^{\,\tw} J^{-1}$.\\
	
	In the same way, taking the $\tw$-unitary operator $U_\tw=u_\tw Ju _\tw J^{-1}=u_\tw\hat{u}_\tw=\hat{u}_\tw u_\tw $, with $u_\tw$ and $\hat{u}_\tw\defeq Ju_\tw J^{-1}$ $\tw$-unitary operators, we get the fluctuation term
	\begin{align*}
		U_\tw\KDir U_\tw^\Tadj&=u_\tw\hat{u}_\tw \KDir \hat{u}_\tw^\Tadj u_\tw^\Tadj =\hat{u}_\tw(u_\tw\KDir u_\tw^\Tadj)\hat{u}_\tw^\Tadj = \hat{u}_\tw(\KDir+ u_\tw [\KDir,u_\tw^\Tadj])\hat{u}_\tw^\Tadj\\
		&= \hat{u}_\tw\KDir \hat{u}_\tw^\Tadj + \hat{u}_\tw \hat{u}_\tw^\Tadj u_\tw [\KDir,u_\tw^\Tadj] \qquad \text{(first-order condition.)}\\
		&= Ju_\tw J^{-1}\KDir J u_\tw^\Tadj J^{-1} +  u_\tw [\KDir,u_\tw^\Tadj]\\
		&= \epsilon^{\prime}\epsilon^{\prime\prime\prime} Ju_\tw J^{-1}J\KDir  u_\tw^\Tadj J^{-1} +  u_\tw [\KDir,u_\tw^\Tadj]\qquad \text{(via $J\KDir= \epsilon^\prime\epsilon^{\prime\prime\prime}\KDir J$)}\\
		&= \epsilon^{\prime}\epsilon^{\prime\prime\prime}J(\KDir+u_\tw[\KDir,u_\tw^\Tadj])J^{-1}+  u_\tw [\KDir,u_\tw^\Tadj]\\
		&=\KDir+ u_\tw [\KDir,u_\tw^\Tadj] +\epsilon^{\prime}\epsilon^{\prime\prime\prime} J u_\tw[\KDir,u_\tw^\Tadj] J^{-1}.
	\end{align*}
	so that fluctuations of $\KDir$ take the form $\KDir_{A^{\,\tw}}=\KDir+ A^{\,\tw} +\epsilon^{\prime}\epsilon^{\prime\prime\prime} J A^{\,\tw} J^{-1}$.
\end{proof}
\newpage
We can now present the main result of this paper.
\begin{theorem}[The connection theorem]
	\label{Thmconnection}
	Let consider the four kind of generalized spectral triples presented in definition \ref{Def4kindsST}. There exists a canonical bijective correspondence between the set of twisted spectral triple and the set of untwisted spectral triple structure provided by the following two connections:
	\begin{itemize}
		\item \textbf{connection 1:} a bijective correspondence between the $ST$ and the $\tw$-$TPRST$;
		\item \textbf{connection 2:} a bijective correspondence between the $\tw$-$TST$ and the $\tw$-$PRST$.
	\end{itemize}
\end{theorem}
 
\begin{proof}
By construction, the generalized spectral triples have the same algebra, real structure and grading operators together with inner product structures related by $\tw$. It remains to present the connection for derivations, boundedness of commutators, first-order conditions, one-forms and fluctuations of the Dirac operators as done below.\\
	
\noindent $\bullet$ \textit{\textbf{Step 1 : Connection for derivations.}} There is a connection between twisted derivations and derivations provided by the bijective correspondence given by  $ [\KDir, a]_\rho=\tw [\Dir, a]$ for connection $1$ and $ [\Dir, a]_\rho=\tw [\KDir, a]$ for connection $2$.\\

\noindent\textit{proof:} Since $\Dir= \tw \KDir $ we have that
\begin{align*}
	& [\KDir, a]_\rho  =\tw\Dir a-\tw a\tw \tw \Dir =\tw [\Dir, a] \qquad\qquad &\text{for connection 1}\\
	&  [\Dir, a]_\rho =\Dir a-\tw a \tw \Dir=   \tw [\KDir, a] \qquad\qquad &\text{for connection 2}
\end{align*}
$\forall a\in\calA.$ This relation is bijective as $\tw$ is unitary.\\

\noindent $\bullet$ \textit{\textbf{Step 2 : mutual implication of boundedness.}} The boundedness of the commutator $[\KDir, a]_\rho$ for the $\tw$-$TPRST$ and the one of $[\Dir, a]$ for the $ST$ implies each other (connection 1). The same holds for $ [\Dir, a]_\rho$ in the $\tw$-$TST$ and $[\KDir, a]$ in the $\tw$-$PRST$ (connection 2).\\

\noindent\textit{proof:} Since $ [\KDir, a]_\rho=\tw [\Dir, a]$ with $\tw$ being a unitary (bounded) operator, the boundedness of $[\KDir, a]_\rho$ and the one of $[\Dir, a]$ implies each other for connection 1. In the same way, noting that $ [\Dir, a]_\rho=\tw [\KDir, a]$, the same holds for connection 2.\\

\noindent $\bullet$ \textit{\textbf{Step 3 : mutual implication of first-order conditions.}} The first-order condition (equation \eqref{FirstOrd}) and the twisted first-order condition (equation \eqref{TwFirstorder}) imply each other in each cases.\\

\noindent\textit{proof:} By a direct computation, we have that:
\begin{align*}
	&[[\KDir, a]_\rho, b^\circ]_{\rho^\circ}=\tw [[\Dir, a], b^\circ]\qquad\qquad &\text{for connection 1}\\
	&[[\KDir, a], b^\circ]=\tw [[\Dir, a]_\rho, b^\circ]_{\rho^\circ}\qquad\qquad &\text{for connection 2}
\end{align*}
so that each first-order condition implies the dual one, thanks to the fact that $\tw^2=\bbbone$.\\

\noindent $\bullet$ \textit{\textbf{Step 4 : Connection for one-forms.}} There is a connection between twisted one-forms $A_\rho^{\,\tw}$ and one-forms $A$ and also between one-forms $A^{\,\tw}$ and twisted one-forms $A_\rho$ given by the bijective relation $A_\rho^{\,\tw}=\tw A$ for connection 1 and $A^{\,\tw}=\tw A_\rho$ for connection 2.\\

\noindent\textit{proof:} Noting that any twisted one-form $A_\rho^{\,\tw}$ can be written as $A_\rho^{\,\tw}=\sum_i\rho(a_i)[\KDir, b_i]_\rho$ for $a_i,b_i\in\calA $, we then have  $A_\rho^{\,\tw}=\sum_i\rho(a_i)[\KDir, b_i]_\rho=\tw \sum_i a_i [\Dir, b_i]\defeq \tw A$ with $A=\sum_i a_i [\Dir, b_i]$ a one-form. In the same way for connection 2 we have $A^{\,\tw}=\sum_ia_i[\KDir, b_i]=\tw\sum_i \rho(a_i)[\Dir, b_i]_\rho=\tw A_\rho$. These relations are bijective.\\

\noindent $\bullet$ \textit{\textbf{Step 5 : Connection for the fluctuations of the Dirac operators.}} There is a connection between twisted fluctuations $\KDir_{A_\rho^{\,\tw}}=\KDir+ A_\rho^{\,\tw} +\epsilon^{\prime}\epsilon^{\prime\prime\prime} J A_\rho^{\,\tw} J^{-1}$ and untwisted fluctuation $\Dir_{A}=\Dir+ A +\epsilon^{\prime} J A J^{-1}$ (connection 1), and also the fluctuation $\KDir_{A^{\,\tw}}=\KDir+ A^{\,\tw} +\epsilon^{\prime}\epsilon^{\prime\prime\prime} J A^{\,\tw} J^{-1}$ and the twisted fluctuations $\Dir_{A_\rho}=\Dir+ {A_\rho} +\epsilon^{\prime} J {A_\rho} J^{-1}$ (connection 2) given by the bijective correspondences $\KDir_{A_\rho^{\,\tw}}=\tw \Dir_{A}$ and $\KDir_{A^{\,\tw}}=\tw \Dir_{A_\rho}$.\\

\noindent\textit{proof:} Thanks to proposition \ref{DualDir} and step $4$, it remains to prove the connection for the last terms of each fluctuation:
\begin{align*}
	& 	\epsilon^{\prime}\epsilon^{\prime\prime\prime} J A_\rho^{\,\tw} J^{-1}   = 	\epsilon^{\prime}\epsilon^{\prime\prime\prime} J\tw A J^{-1}  = \tw	\epsilon^{\prime}   J A J^{-1}   \qquad&\text{for connection 1}\\
	& 	\epsilon^{\prime}\epsilon^{\prime\prime\prime} J A^{\,\tw} J^{-1}   = 	\epsilon^{\prime}\epsilon^{\prime\prime\prime} J\tw A_\rho J^{-1}  =\tw	\epsilon^{\prime}   J A_\rho J^{-1}  \qquad&\text{for connection 2}
\end{align*} 
where the explicit form for $A_\rho^{\,\tw}$, $A^{\,\tw}$, $A_\rho$ and $A$ is taken to be the same as in step $4$. These relations are bijective.
\end{proof}
 
We are now interested in the way this connection manifests for the inner fluctuations of the Dirac operators.

Let's start with a useful technical proposition.
\begin{lemma}
	\label{PropTechUnit}
	If an operator $O$ is unitary (resp $\tw$-unitary), then $\rho(O)$ and $OJOJ^{-1}$ are unitary (resp $\tw$-unitary) operators.
\end{lemma}
\begin{proof}
If $O$ is unitary then $\bbbone=\rho(OO^\dagger)=\rho(O)\rho(O)^\dagger$ and $\bbbone=\rho(O^\dagger O)=\rho(O)^\dagger\rho(O)$ hence $\rho(O)$ is unitary. In the same way $(OJOJ^{-1})(OJOJ^{-1})^\dagger=(OJOJ^{-1})JO^\dagger J^{-1}O^\dagger=\bbbone$ and $(OJOJ^{-1})^\dagger(OJOJ^{-1})=\bbbone$ so that $(OJOJ^{-1})$ is also unitary.

If $O$ is $\tw$-unitary then $\bbbone=\rho(OO^\Tadj)=\rho(O)\rho(O)^\Tadj$ and $\bbbone=\rho(O^\Tadj O)=\rho(O)^\Tadj\rho(O)$ hence $\rho(O)$ is $\tw$-unitary. Using that $\,\rho( J)=\epsilon^{\prime\prime\prime}J$ we have that $(OJOJ^{-1})(OJOJ^{-1})^\Tadj=(\epsilon^{\prime\prime\prime})^2(OJOJ^{-1})JO^\Tadj J^{-1}O^\Tadj=\bbbone$ and then $(OJOJ^{-1})^\Tadj(OJOJ^{-1})=\bbbone$ which implies that $(OJOJ^{-1})$ is $\tw$-unitary.
\end{proof}

A preliminary general result concerning fluctuations is provide by the following.
\begin{lemma}[Connection for the fluctuations of the Dirac operators]
\label{PropEqUnit} 	
There is a connection between the fluctuations of the Dirac operators in each dual contexts. For connection 1 it is given by $U\KDir U^\Tadj=\tw V \Dir V^\dagger$ with $U$ and $V=\rho(U)$ two unitary operators in $\calB(\calH)$ or by the relation $U_\tw\KDir U_\tw^\dagger=\tw V_\tw \Dir V_\tw^\Tadj$ with $U_\tw$ and $V_\tw=\rho(U_\tw)$ being two $\tw$-unitary operators in $\calB(\calH)$. For connection 2 it is given by $U\KDir U^\dagger=\tw V  \Dir V^\Tadj$ and $U_\tw \KDir U_\tw^\Tadj=\tw V_\tw \Dir V_\tw^\dagger$. 
\end{lemma}

\begin{proof}
	Using that $O\tw  = \tw\rho(O)$ for any operator $O\in\calB(\calH)$, for connection 1 we have
	\begin{align}
		\label{D11}
		U\KDir U^\Tadj=U\tw\Dir U^\Tadj=\tw \rho(U)\Dir U^\Tadj=\tw V \Dir  V^\dagger
	\end{align}
	where $V\defeq \rho(U)$ is unitary thanks to lemma \ref{PropTechUnit}. For the second kind of action:
	\begin{align}
		\label{D12}
		U_\tw\KDir U_\tw^\dagger=U_\tw\tw\Dir U_\tw^\dagger=\tw \rho(U_\tw)\Dir U_\tw^\dagger=\tw V_\tw \Dir  V_\tw^\Tadj 
	\end{align}
	since lemma \ref{PropTechUnit} implies that $V_\tw\defeq \rho(U_\tw)$ is also a $\tw$-unitary operator. In the same way, for connection 2 we have
	\begin{align}
		\label{D21}
		&	U  \KDir U^\dagger=\tw \rho(U) \Dir U^\dagger=\tw V  \Dir V^\Tadj\\
		\label{D22}
		&	U_\tw \KDir U_\tw^\Tadj=\tw \rho(U_\tw) \Dir U_\tw^\Tadj=\tw V_\tw \Dir V_\tw^\dagger.
	\end{align}
\end{proof}
Now focusing on inner fluctuations, taking an element $a\in\calA$ and using the relation$\rho( J)=\epsilon^{\prime\prime\prime}J$ we have  that $\rho(Ad(a))= Ad(\rho(a))=\rho(a)J\rho(a)J^{-1}=\rho(a)\rho^\circ(JaJ^{-1})$. Taking the unitary operator $U=Ad(u)=uJuJ^{-1}$ with $u$ a unitary operator in $\calA$ or the $\tw$-unitary operator  $U_\tw=Ad(u_\tw)=u_\tw Ju_\tw J^{-1}$ with $u_\tw$ a $\tw$-unitary operator in $\calA$, there are different ways to generate the fluctuations in each of the four contexts:
\begin{enumerate}
	\item By the actions $U\Dir U^\dagger$ or $U_\tw \Dir U_\tw^\Tadj$ to generate $\Dir_{A}$ for $ST$.
	\item By the actions $U\Dir U^\Tadj$ or $U_\tw \Dir U_\tw^\dagger$ to generate $\Dir_{A_\rho}$ for $\tw$-$TST$.
	\item By the actions $U\KDir U^\dagger$ or $U_\tw \KDir U_\tw^\Tadj$ to generate $\KDir_{A^{\,\tw}}$ for $\tw$-$PRST$.
	\item By the actions $U\KDir U^\Tadj$ or $U_\tw \KDir U_\tw^\dagger$ to generate $\KDir_{A_\rho^{\,\tw}}$ for $\tw$-$TPRST$.
\end{enumerate}

\begin{proposition}[Connection at the level of inner fluctuations] There is a connection between the inner actions that generate fluctuations. For connection 1 it is given by $U\KDir U^\Tadj=\tw V \Dir V^\dagger$ with $U=uJuJ^{-1}$ and $V=\rho(U)=\rho(u)J\rho(u)J^{-1}$ another unitary or by the relation $U_\tw\KDir U_\tw^\dagger=\tw V_\tw \Dir V_\tw^\Tadj$ with $U_\tw=u_\tw Ju_\tw J^{-1}$ and $V_\tw=\rho(U_\tw)=\rho(u_\tw)J\rho(u_\tw)J^{-1}$ another $\tw$-unitary operator. For connection 2 it is given by $U\KDir U^\dagger=\tw V  \Dir V^\Tadj$ and $U_\tw \KDir U_\tw^\Tadj=\tw V_\tw \Dir V_\tw^\dagger$. 
 \end{proposition}
 \begin{proof}

The proof that $U\KDir U^\Tadj=\tw V \Dir  V^\dagger$ has been given in lemma \ref{PropEqUnit} we then have
\begin{align}
	\label{DI11}
\KDir_{A_\rho^{\,\tw}}\defeq U\KDir U^\Tadj=\tw V \Dir  V^\dagger\defeq \tw \Dir_A.
\end{align}
In the same way, in the other cases we have:
\begin{align}
	\label{DI12}
	&\KDir_{A_\rho^{\,\tw}}\defeq U_\tw\KDir U_\tw^\dagger=\tw V_\tw \Dir  V_\tw^\Tadj\defeq \tw \Dir_A\\
	\label{DI21}
	&\KDir_{A^{\,\tw}}\defeq	U  \KDir U^\dagger=\tw V  \Dir V^\Tadj\defeq\tw \Dir_{A_\rho}\\
	\label{DI22}
	&\KDir_{A^{\,\tw}}\defeq	U_\tw \KDir U_\tw^\Tadj=\tw V_\tw \Dir V_\tw^\dagger\defeq\tw \Dir_{A_\rho}.
\end{align}
 \end{proof}
Notice that the connection can be equivalently written as an equality of the evaluation of the fluctuated Dirac operators by their respective inner products:
\begin{align*}
\langle \psi ,\KDir_{A_\rho^{\,\tw}} \psi\rangle_\tw = \langle \psi ,  \Dir_{A}  \psi\rangle\qquad\qquad\qquad&\text{for connection 1}\\
\langle \psi ,\KDir_{A^{\,\tw}} \psi\rangle_\tw = \langle \psi , \Dir_{A_\rho} \psi\rangle\qquad\qquad\qquad&\text{for connection 2}
\end{align*}
highlighting the fact that the heart of the connection lies in the fact that this is the relation between the algebraic structures (inner products, Dirac operators, derivations..) that is conserved, going from one generalized spectral triple to its dual one. 

 \begin{proposition}[Axioms for the dual $\tw$-$PRST$ and $\tw$-$TPRST$]
 	\label{PropAxDual}
Starting from the usual algebraic relations defining $ST$ and $\tw$-$TST$, i.e. $J^2 = \epsilon$, $J \Dir = \epsilon^{\prime} \Dir J$, $J \Gamma = \epsilon^{\prime\prime} \Gamma J$, $\Dir \Gamma+\Gamma\Dir=0$, $\Gamma \pi(a) = \pi(a)\Gamma$ for any $a\in\calA$ and $J\tw=\epsilon^{\prime\prime\prime}\tw J$, then the corresponding algebraic constraints for the duals $\tw$-$PRST$ and $\tw$-$TPRST$ are given by
\begin{equation}
J\KDir= \epsilon_\tw^\prime\KDir J \qquad\text{and}\qquad \KDir\Gamma+\rho(\Gamma)\KDir=0
\end{equation}
with $\epsilon_\tw^\prime=\epsilon^\prime\epsilon^{\prime\prime\prime}$, the other relations being unchanged.
 \end{proposition}
\begin{proof}

The relation  $J\Dir=\epsilon^\prime \Dir J$ for usual spectral triple implies $J\KDir=J\tw \Dir=\epsilon^{\prime\prime\prime}\tw J \Dir= \epsilon^\prime\epsilon^{\prime\prime\prime}\KDir J$ and we have that $\KDir\Gamma=\tw \Dir \Gamma=-\tw \Gamma \Dir =-\rho(\Gamma)\tw \Dir =-\rho(\Gamma)\KDir$.
\end{proof}

\begin{proposition}
In the case of the twist by the grading, where the algebra is taken to be $\calA\otimes \mathbb{C}^2$, we have that $\Gamma\KDir=\KDir\Gamma$.
\end{proposition}
\begin{proof}
Since $\Gamma$ is an element of $\calA\otimes \mathbb{C}^2$ that can be written as $\Gamma=p_+-p_-\equiv (\bbbone, -\bbbone)$ we have that $\rho(\Gamma)=-\Gamma$, hence the result.
\end{proof}
\begin{remark}
	\label{RQExchange}
	Interestingly, the relation $\Gamma \Dir + \Dir \Gamma = 0$ implies that $\Dir$ exchange the $\calH_\pm$ spaces, i.e. $\Dir\, :\, \calH_\pm\, \to \, \calH_\mp$, and the corresponding relation $\Gamma\KDir-\KDir\Gamma=0$ for the twist by the grading implies that $\KDir$ preserve these to spaces $\KDir\, :\, \calH_\pm\, \to \, \calH_\pm$.
\end{remark}

\begin{lemma}
Any $\tw$-$TPRST$ can be associated with a dual $ST$ and vice versa. In the same way, any $\tw$-$PRST$ can be associated with a dual $\tw$-$TST$ and vice versa. The connections between Dirac operators and inner products are given by equations \eqref{EqDualD} and \eqref{Eqtwprod}.  
\end{lemma}

An important consequence for physics is provided by the following corollary.
\begin{corollary}[Connection for the fermionic and spectral actions]
	\label{PropDualFermAct}
The usual fermionic action evaluated on $\Dir_A$ and the $\tw$-twisted fermionic action evaluated on $\KDir_{A_\rho}=\tw \Dir_{A}$ are equal. The corresponding spectral actions are also equal.
\end{corollary}
 \begin{proof}
A direct computation gives
\begin{align}
&\act_f^\tw( \KDir_{A_\rho^{\,\tw}}, \psi)=\langle \psi , \KDir_{A_\rho^{\,\tw}} \psi \rangle_\tw=  \langle \psi , \tw \tw \Dir_A \psi \rangle=  \act_f( \Dir_A, \psi)\\
&\act[\KDir_{A_\rho^{\,\tw}}]= \Tr f(\KDir_{A_\rho^{\,\tw}} (\KDir_{A_\rho^{\,\tw}})^\dagger / \Lambda^2)= \Tr f(\tw\Dir_{A} \Dir_{A}^\dagger \tw^\dagger  / \Lambda^2)=\act[\Dir_{A}]
\end{align}
 \end{proof}

This imply that the usual spectral invariant are conserved when going from one generalized spectral triple to the dual one. This is an important aspect of the connection.   

The axiom $\Gamma \Dir + \Dir \Gamma = 0$ for spectral triple will be generalized to $\Gamma \Dir \pm \Dir \Gamma = 0$ in subsection \ref{SubSecDiffST} so that the corresponding relation for $\KDir$ will be given by $\Gamma\KDir\mp\KDir\Gamma=0$.
 
\subsection{$\tw$-morphism of generalized spectral triples}
\label{SubSecTwMorph}

Let $\calI_i\defeq (V, ( \, .\, , \, .\, )_i)$ for $i=1,2$ be general (indefinite or definite) inner product spaces with $V$ a vector space as previously. The corresponding adjoint will be denoted by $\dagger_i$. We can define two (unspecified) Generalized Spectral Triples $GST_i\defeq(\calA_i, \calI_i, \calD_i, \tw_i)$ where $\calA_i$'s are involutive algebras acting as bounded operator algebras on $\calI_i$. Taking $a\in\calA_i$, the bounded commutator is given by $[\Dir_i, a]_{\rho_i}$ where $\rho_i(\,.\,)\defeq \tw_i(\,.\,)\tw_i$ is a fundamental twist for $i=1,2$. Such generalized spectral triples contain any of the four kinds of generalized spectral triples defined in subsection \ref{SubSecPresDual} as a particular case. The specific notation $(\calA_i, \calI_i, \calD_i)$ corresponds to the case where $\rho_i$ is a trivial twist (on $\calA_i$). We can now define a particular morphism that will permit to connect different generalized spectral triples by means of this transformation.     
 \begin{definition}[$\tw$-morphism of generalized spectral triple]
Given a fundamental twist $\rho$ implemented by the fundamental symmetry $\tw$ i.e. $\rho(\, .\,)=\tw (\, .\,) \tw$, a $\tw$-morphism of generalized spectral triples is the map $\twphi\, :\, GST_1\, \to\, GST_2$ defined by
\begin{align*}
&\calA_1\;\to\, \calA_2\defeq \rho(\calA_1)\qquad\quad &\calI_1\,\to\, \calI_2\defeq (V, ( \, .\, , \, .\, )_2\defeq ( \, .\, , \tw\, .\, )_1)\\
&\calD_1\, \to\, \calD_2\defeq \tw \Dir_1\qquad\quad &[\Dir_1, a]_{\rho_1}\,\to\, [\Dir_2, a]_{\rho_2}\defeq \tw[\Dir_1, a]_{\rho_1}
\end{align*}
with $\rho_2=\rho\,\circ \rho_1$ or equivalently $\rho_2(\,.\,)\defeq \tw\tw_1(\,.\,)\tw_1\tw$.
 \end{definition}
As $\tw$ is a fundamental symmetry, the fact that $(\tw)^2=\bbbone$ implies that the $\tw$-morphism is involutory, i.e. $\twphi\circ \twphi=\phi^{\, (\tw)^2}=\phi^{\, \bbbone}$ which is the identity morphism. We then have that the inverse of $\twphi$ always exists and is given by $\twphi$ itself, inducing that $\twphi$ is a morphism.
\begin{remark}
	\label{RqTwRelatDef}
The requirement that $\rho$ is a fundamental twist is necessary in order that any (indefinite or not) inner product will be transformed by the $\tw$-morphism in such another inner product (without the complex phase relating the two inner products in equation \eqref{EqPhaseRelaInners}), thanks to a similar demonstration as the one of proposition \ref{PropTwistProd}.
\end{remark}
 
\begin{lemma}
If $\calD_1$ is selfadjoint for $( \, .\, , \, .\, )_1$ then the additional requirement that $\tw=\tw^{\dagger_1}$ implies that the resulting Dirac operator $\calD_2=\tw\calD_1 $ is selfadjoint for $( \, .\, , \, .\, )_2=( \, .\, , \tw\, .\, )_1$.
\end{lemma}
\begin{proof}
Taking $\psi\in V$ we have that $(\psi,\calD_2 \psi)_2=(\psi,\calD_1 \psi)_1=(\calD_1\psi, \psi)_1=( (\tw)^2\calD_1\psi, \psi)_1=( \tw \calD_1\psi, \tw^{\dagger_1}\psi)_1=( \tw \calD_1\psi, \tw\psi)_1=( \calD_2\psi,  \psi)_2$.
\end{proof}
This proposition generalises proposition \ref{DualDir} where the inner product $\langle \, .\, , \, .\, \rangle$ is replaced here by $( \, .\, , \, .\, )_1$ and the requirement $\tw=\tw^\dagger$ by $\tw=\tw^{\dagger_1}$.
\begin{lemma}
The composition of two $\tw$-morphisms $\phi^{\,\tw_1}$ and $\phi^{\,\tw_2}$ is equal to $\phi^{\,\tw_1}\circ\phi^{\,\tw_2}=\phi^{\,\tw_1\tw_2}$ which is also a $\tw$-morphism only if $[\tw_1,\tw_2]=0$. 
\end{lemma}
\begin{proof}
The fact that $\phi^{\,\tw_1}\circ\phi^{\,\tw_2}=\phi^{\,\tw_1\tw_2}$ can directly be seen from the definition of a $\tw$-morphism. We have that $\tw_1\tw_2$ is unitary since $\tw_1$ and $\tw_2$ are also unitaries. We must now check to which condition $\rho_{1,2}(\, .\,)\defeq \tw_1\tw_2 (\, .\,)\tw_2 \tw_1$ is regular. Thanks to lemma \ref{PropRegImpFund}, we know that $\rho_{1,2}$ is a fundamental twist only if $\tw_1\tw_2=\tw_2\tw_1$, hence the result.
\end{proof}
The interest behind the concept of the $\tw$-morphism lies in the fact that they transform a given generalized triple (as introduced in subsection \ref{SubSecPresDual}) into its dual one. For connection 1 we have that
\begin{itemize}
	\item $\twphi\, :\, ST\, \to\, \tw\text{-}TPRST$:  $\twphi(\calA,\calH, \Dir, J, \Gamma )=(\calA,\calK_\tw, \KDir, J, \Gamma, \tw )$.
		\item $\twphi\, :\, \tw\text{-}TPRST\, \to\, ST$  : $\twphi(\calA,\calK_\tw, \KDir, J, \Gamma, \tw )=(\calA,\calH, \Dir, J, \Gamma )$.
		\end{itemize}
And for connection 2:
			\begin{itemize}
	\item $\twphi\, :\, \tw\text{-}PRST\, \to\, \tw\text{-}TST$ :  $\twphi(\calA,\calK_\tw, \KDir, J, \Gamma )=(\calA,\calH, \Dir, J, \Gamma, \tw )$.
	\item $\twphi\, :\, \tw\text{-}TST\, \to\, \tw\text{-}PRST$ : $\twphi(\calA,\calH, \Dir, J, \Gamma, \tw )=(\calA,\calK_\tw, \KDir, J, \Gamma )$.
	\end{itemize}
 Where we use that $\calA\simeq \rho(\calA)$ as $\rho$ is an automorphism. These transformations are bijective since $\twphi$ is a morphism. Corollary \ref{PropDualFermAct} imply that the $\tw$-morphism is not merely an algebraic morphism but an invariance of the physical actions.

\newpage

\section{Connection for even-dimensional manifolds}
\label{SecMFLD}
 
We are now interested in the way the connection of the previous section manifests itself in the case of the spectral triple of even-dimensional smooth compact manifolds, taking $\calA\defeq C^\infty(\Man)\otimes \mathbb{C}^2$\footnote{Taking $\calA\defeq C^\infty(\Man)\otimes \mathbb{C}^2$ is nescessary to present the results in full generality, but we will see that the restriction to $C^\infty(\Man)$ correspond to the case of interest}. As shown in \cite{TwistLandiMarti2016}, the twist by the grading is the only way to minimally twist this spectral triple, making this model the canonical way to implement the twisting procedure in this case. We will therefore consider the model of the minimal twist of the manifold in this section, as done in \cite{TwistLandiMarti2016, martinetti2024torsion}.

Following \cite{strohmaier2006noncommutative}, we will keep the standing assumption that $\Man$ is a smooth compact manifold in this section. In particular, when the pseudo-Riemannian metric $\gpr$ has Lorentzian signature, $(\Man,\gpr)$ should not be interpreted as a realistic globally hyperbolic spacetime, since compact Lorentzian manifolds are known to contain closed timelike curves and hence to be acausal (see for instance \cite{galloway1984splitting}). The constructions below are therefore meant as local, algebraic models at the level of Clifford and Dirac data, illustrating how the twist implements a change of signature, rather than as a full treatment of Lorentzian 
space-times with their global causal structure.

In subsection \ref{SubsecWickPar}, we present the two contexts that will be used to present the connection. The first context is based on a Riemannian metric with associated Clifford algebra (resp pseudo-Riemannian metric with associated Clifford algebra for the second context). We then show how special kinds of fundamental symmetries permit the construction of a generalized parity operator, permitting the signature change. In subsection \ref{SubSecDiffST} we show in which cases these fundamental symmetries can be connected with the twist by the grading, and what are the potential generalized spectral triples that we can build from this connection, in order to set up the connection. In subsection \ref{SubSectDualMet} we show how the corresponding $\tw$-morphism implements the particular signature change described in subsection \ref{SubsecWickPar}, and we introduce the notion of twisted Clifford algebra. The particular 4-dimensional compact pseudo-Riemannian model with Lorentzian signature $(1,3)$ is presented in Subsection~\ref{SubSec4dlor}.

 \subsection{Signature change through the parity operator}
\label{SubsecWickPar}
 
We consider a $2m$-dimensional smooth compact manifold $\Man$ associated with a non degenerate metric $\g$ of signature $(n, 2m-n)$. A local oriented orthonormal basis of $T\Man$ is given by $\{E_1, \dots, E_{2m}\}$ so that the metric can be written accordingly as $g_{ab}=\g(E_a, E_b)$. The dual basis is given by $\{\theta^1, \dots, \theta^{2m}\}$ with the defining relation $\g(\theta^b, E_a)=\delta^b_a$. The Clifford representation associated with $\g$ is generated by the elements $\cl(v)$ for $v\in T_x\Man$ so that the corresponding gamma matrices are given by the unitary operators $\gamma^a\defeq  \cl(\theta^a)$ for any $a\in\{1,\dots, 2m\}$, with the defining relations $\{\gamma^a, \gamma^b\}=2\g^{ab}\bbbone_{2^{m}}$. The corresponding grading is defined by $\Gamma=i^{-m(2m-1)-n} \prod_{a=1}^{2m} \gamma^a$ and satisfy $\Gamma^2=\bbbone$ and $\Gamma^\dagger=\Gamma$. 

The metric $\g$ induces the splitting $T\Man=T\Man_\g^+\oplus T\Man_\g^-$, with the local oriented orthonormal basis given by $\{E_1, \dots, E_{n}\}$ for $T\Man^+$ and $\{E_{n+1}, \dots, E_{2m}\}$ for $T\Man^-$ so that any elements $v\in T_x\Man$ take the form $v=v^+\oplus v^-$ with $\g(v^+,v^+)\geq 0$ and $\g(v^-,v^-)\leq 0$. The signature of $\g$ is then given by $n$ plus signs and $2m-n$ minus signs according to this basis.

Now we can consider the algebra $\calA\defeq C^\infty(\Man)\otimes \mathbb{C}^2$ and its representation according to $\Gamma$ as done in equation \eqref{EqRhogen}. In this way, using the projections $p_\pm:=  (\bbbone \pm \Gamma)/2$ so that $\calH=p_+\calH\oplus p_-\calH:=\calH_+\oplus \calH_-$, all elements $a=(f,g)\in C^\infty(\Man)\otimes \mathbb{C}^2$ will be represented accordingly by the chiral representation $\pi_{\,\Gamma\,}(f,g)$ induced by the grading so that
\begin{equation}
	\label{EqRho}
	\pi_{\,\Gamma\,}(f,g)\defeq p_+\pi_0(f)+p_-\pi_0(g)=	\begin{pmatrix}
		\pi_+(f) & 0  \\
		0 & \pi_-(g)  \\
	\end{pmatrix}
\end{equation} 
with $\pi_0$ the representation of $C^\infty(\Man)$ and $\pi_\pm$ its restrictions to $\calH_\pm$, as defined in \cite{TwistLandiMarti2016}.

\begin{proposition}
For any $a\in\{1, \dots, 2m\}$, the gamma matrices $\gamma^a$ take the following form:
\begin{equation}
	\label{EqGammaForm}
	\gamma^a =
	\left( \begin{array}{cc}
		0 & \sigma^a \\ \tilde\sigma^a & 0
	\end{array} \right)
\end{equation}
in the chiral representation $	\pi_{\,\Gamma\,}$, with $\sigma^a \tilde\sigma^b +\sigma^b\tilde\sigma^a=2 \g^{ab}\bbbone_{2^{m-1}}$.
\end{proposition}
\begin{proof}
Since the dimension is $2m$, we have the relation $\gamma^a\Gamma=-\Gamma\gamma^a$ which induces their off-diagonal form according to $    \pi_{\,\Gamma\,}$. Then the relation $\{\gamma^a, \gamma^b\}=2\g^{ab}\bbbone_{2^{m}}$ implies from a direct computation that $\sigma^a \tilde\sigma^b +\sigma^b\tilde\sigma^a=2 \g^{ab}\bbbone_{2^{m-1}}$.
\end{proof}
 In the following, we will work with two different contexts to show how the connection present in each of the two cases. We start from a manifold $\Man$ of dimension $2m$ by considering two metrics, $\gr$ and $\gpr$ being respectively positive definite and indefinite. The first context deals with the Clifford representation associated with $\gr$ and the second context deals with the one of $\gpr$. We impose for simplicity that the two metrics $\gr$ and $\gpr$ admit the same local oriented orthonormal basis for $T\Man$, given by $\{E_1, \dots, E_{2m}\}$ as for $\g$. We can now present the specific notations that will be used in both contexts:
\begin{itemize}
	\item \textbf{Context 1: }Consider a Riemannian manifold $(\Man, \gr)$ of dimension $2m$ with tangent bundle $T\Man$ and non degenerate metric $\gr$ of signature $(2m, 0)$. A local oriented orthonormal basis of $T\Man$ is given by $\{E_1, \dots, E_{2m}\}$ so that the metric can be written accordingly as $g^{\scriptscriptstyle R}_{ab}=\gr(E_a, E_b)$. The dual basis $\{\theta_{\scriptscriptstyle R}^1, \dots, \theta_{\scriptscriptstyle R}^{2m}\}$ is defined by the relation $\gr(\theta_{\scriptscriptstyle R}^b, E_a)=\delta^b_a$. The associated Clifford algebra $Cl_{2m}$ is generated by the elements $\clr(v)$ with $v\in T_x\Man$ so that we define the usual gamma matrices $\{\gmr^a \}_{a\in\{1, \dots, 2m\}}$ with $\gmr^a \defeq \clr(\theta_{\scriptscriptstyle R}^a)$ for any $a\in\{1, \dots, 2m\}$ where $\clr$ is chosen so that $\gmr^a$ is unitary, with the defining relations $\{\gmr^a, \gmr^b\}=2\gr^{ab}\bbbone_{2^{m}}$. The Dirac operator associated with $\clr$ is given by $-i\gmr^a\nabla_a^{{\scriptscriptstyle R, S}}$ where $\nabla_a^{{\scriptscriptstyle R, S}}$ is the induced spin connection, see \cite{gracia2013elements} for more details. The grading is defined by $\Gamma_{\scriptscriptstyle R}=i^{m} \prod_{a=1}^{2m} \gmr^a$ and satisfy  $\Gamma_{\scriptscriptstyle R}^2=\mbb$ and $\Gamma_{\scriptscriptstyle R}^\dagger=\Gamma_{\scriptscriptstyle R}$. 
	\item \textbf{Context 2: }Consider the pseudo-Riemannian manifold $(\Man, \gpr)$ of dimension $2m$ together with the non degenerate metric $\gpr$ of signature $(n, 2m-n)$ with $0<n<2m$. The metric $\gpr$ induces the splitting $T\Man=T\Man^+\oplus T\Man^-$, with the local oriented orthonormal basis being given by $\{E_1, \dots, E_{n}\}$ for $T\Man^+$ and $\{E_{n+1}, \dots, E_{2m}\}$ for $T\Man^-$ so that any elements $v\in T_x\Man$ take the form $v=v^+\oplus v^-$ such that $\gpr(v^+,v^+)\geq 0$ and $\gpr(v^-,v^-)\leq 0$. The signature of $\gpr$ is then given by $n$ plus signs and $2m-n$ minus signs according to this basis.  The dual basis $\{\theta_{\scriptscriptstyle PR}^1, \dots, \theta_{\scriptscriptstyle PR}^{2m}\}$ is defined by the relation $\gpr(\theta_{\scriptscriptstyle PR}^b, E_a)=\delta^b_a$. The corresponding Clifford algebra $Cl_{n,\, 2m-n}$ is generated by the elements $\clpr(v)$ with $v\in T_x\Man$ so that we can define the pseudo-Riemannian gamma matrices $\{\gmpr^a \}_{a\in\{1, \dots, 2m\}}$ so that $ \gmpr^a \defeq  \clpr(\theta_{\scriptscriptstyle PR}^a)$ for any $a\in\{1, \dots, 2m\}$ whith $\clpr$ taken so that $\gmpr^a$ is unitary, with the defining relations $\{\gmpr^a, \gmpr^b\}=2\gpr^{ab}\bbbone_{2^{m}}$ with $\g_{ab}^{\scriptscriptstyle PR}=\gpr(E_a, E_b)$. In the same way as for the Riemannian case, we can construct the corresponding Dirac operator $-i\gmpr^a\nabla_a^{{\scriptscriptstyle PR, S}}$ where $\nabla_a^{{\scriptscriptstyle PR, S}}$ is the associated spin connection. The grading $\Gamma_{\scriptscriptstyle PR}$ is given by $\Gamma_{\scriptscriptstyle PR}=i^{-m(2m-1)-n} \prod_{a=1}^{2m} \gmpr^a$ with $\Gamma_{\scriptscriptstyle PR}^2=\mbb$ and $\Gamma_{\scriptscriptstyle PR}^\dagger=\Gamma_{\scriptscriptstyle PR}$.
\end{itemize}

We will impose that $\clr$ and $\clpr$ are chosen so that $\Gamma_{\scriptscriptstyle R}=\Gamma_{\scriptscriptstyle PR}=\Gamma$. This choice is made to avoid the proliferation of specific notations associated with the chiral representation deduced from $\Gamma$. This does not deprive the coming results of their generality.  

 \begin{definition}[Reflection]
Taking $(\Man, \g)$ with $\g$ being any metric, a reflection $r$ is an isometric\footnote{This means that $\g(\, r .\,,\, r .\, )=\g(\, .\,,\, .\, )$.} automorphism on $T\Man$ such that $r^2=\bbbone$.
 \end{definition}
The property $r^2=\bbbone$ implies that $r$ induces the natural splitting $T\Man=T\Man_r^+\oplus T\Man_r^-$ by the two eigenspaces of $r$ corresponding to the eigenvalues $\pm 1$. Then any element $v\in T_x\Man$ can be written accordingly as $v=v_r^+\oplus v_r^-$ with $r v_r^\pm=\pm v_r^\pm$. The two spaces $T\Man_r^\pm$ are orthogonal for $\g$ since $\g(v_r^-, v_r^+)=\g(rv_r^-, rv_r^+)=-\g(v_r^-, v_r^+)$ implies $\g(v_r^-, v_r^+)=0$ for any $v\in T_x\Man$. Then $T\Man_r^+$ admits a local oriented orthonormal basis given by  $\{E_{a_1}, \dots, E_{a_k}\}$ for $T\Man_r^+$ with $k=\dim(T\Man_r^+)$ (resp $\{E_{a_{k+1}}, \dots, E_{a_{2m}}\}$ for $T\Man_r^-$) where $i\in \{1, \dots, 2m\}$ and the indices $a_i\in \{1, \dots, 2m\}$ so that $i\neq j$ induce $a_i\neq a_j$. We define $l$ to be the number of elements of the basis of $T\Man_\g^-$ that are also in the basis of $T\Man_r^+$ i.e. $l=\dim(T\Man_g^-\cap T\Man_r^+)$.

\begin{definition}[Fundamental symmetry associated with a reflection]
	\label{defFundRefl}
	To any reflection operator $r$ can be associated a fundamental symmetry $\calJ_r$ defined by 
\begin{equation}
	\label{EqGenFondSym}
	\calJ_r\defeq i^{-k(k-1)/2-l}\prod_{i=1}^{k}\gamma^{a_i}
\end{equation}
where the indices $\{a_1,\dots, a_k\}$ correspond to the elements of the basis of $T\Man_r^+$. The corresponding indefinite inner product is defined by $\langle \, .\, , \calJ_r\,  .\, \rangle$.  
\end{definition}

The action of the reflection $r$ on the clifford representation is given by the operator $P_r$ defined by the action $P_r: \cl(v)\,\to\, P_r(\cl(v))\defeq \cl(rv)$.
\begin{proposition}
	\label{PropParOp}
The operator $P_r$ is explicitly given by the following inner action of $\calJ_r$:
\begin{equation}
	\label{EqRefleGam}
P_r(\cl(v))=(-1)^{k+1}\calJ_r\cl(v)\calJ_r= \cl(rv)
\end{equation}
\end{proposition}
\begin{proof}
We have that 
\begin{align*}
\cl(rE_{a_i})=\left\{
\begin{array}{ll}
 	\,\,\, 	\cl(E_{a_i})= \gamma_{a_i}&\text{ for }  \quad i\in\{ 1, \dots , k\},\\ 
	 -\cl(E_{a_i})=-\gamma_{a_i}& \text{ for } \quad i\in\{k+1, \dots , 2m\} \end{array}\right.
\end{align*}
 
since $\cl$ is linear. In the same way, using the permutation relations of the gamma matrices and the fact that $\calJ_r^2=\bbbone$ we have
\begin{equation}
(-1)^{k+1}\calJ_r\gamma_{a_i}\calJ_r=\left\{
	\begin{array}{ll}
		\,\,\, \gamma_{a_i} &\text{ for }  \quad i\in\{ 1, \dots , k\},\\ 
		   -\gamma_{a_i}&\text{ for } \quad i\in\{k+1, \dots , 2m\} \end{array}\right.
\end{equation}
hence the equality $(-1)^{k+1}\calJ_r\cl(E_{a_i})\calJ_r=\cl(rE_{a_i})$ so the result of equation \eqref{EqRefleGam} by the linearity of the Clifford representation $\cl$ when writing $v=\sum_{a=1}^{2m}v^aE_a$ with $v^a\in \bbR$.
\end{proof}

\begin{proposition}[Signature change by the reflection operator]
Starting from a metric $\g$ and taking a reflection $r$, a new metric $\g^r$ can be defined by $g^r(\, .\,,\, .\, )\defeq\g(\, .\,,\,r .\, )$ are equivalently at the level of the Clifford algebra taking $v, w\in T\Man$ by
\begin{equation}
g^r(v,w)\mbb\defeq\frac{1}{2}\{\cl(v), \cl(r w)\}=\frac{1}{2}\{\cl(v), P_r(\cl(w))\}. 
\end{equation}
The signature of $g^r$ according to $\{E_a, \dots, E_{2m}\}$ is obtained from the signature of $\g$ where (only) the signs corresponding to the $a_i$'s entries with $i\in\{k+1, \dots , 2m\}$ have been flipped. 
\end{proposition}
\begin{proof}
We have that $g_{a_ia_i}^r\mbb=\frac{1}{2}\{\gamma^{a_i}, P_r(\gamma^{a_i})\}$ is equal to $\frac{1}{2}\{\gamma^{a_i}, \gamma^{a_i}\}=g_{a_ia_i}\mbb$ when $i\in\{ 1, \dots , k\}$ and $-\frac{1}{2}\{\gamma^{a_i}, \gamma^{a_i}\}=-g_{a_ia_i}\mbb$ when $i\in\{k+1, \dots , 2m\}$.
\end{proof}

 \begin{definition}[Spacelike reflection]
 	Taking $(\Man, \gpr)$ to be a pseudo-Riemannian manifold, a spacelike reflection $r$ is a reflection for the metric $\gpr$ so that
 	\begin{equation}
 		\label{EqmetDef}
 		\gr(\, .\,,\, .\, )\defeq \gpr(\, .\,,\,r .\, )
 	\end{equation}
 	is a positive definite metric on $T\Man$.
 \end{definition}
In this case, the splitting $T\Man=T\Man^+\oplus T\Man^-$ coming with $\gpr$ and the one coming with $r$ coincide ($T\Man_r^\pm=T\Man^\pm$). We then have that any vector $v\in T_x\Man$ can be written as $v=v^+\oplus v^-$ with $v^\pm\in T_x\Man^\pm$ so that $r v^\pm=\pm v^\pm$. The passage from $\gr$ to $\gpr$ with the spacelike reflection can be considered as a way to implement the signature change.

From equation \eqref{EqGenFondSym} we can now define two fundamental symmetries implementing the spacelike reflection $r$, according to the two contexts:
 \begin{equation}
	\label{eqSymCr}
	\fnr\defeq i^{-n(n-1)/2}\gmr^1\dots \gmr^n\qquad\qquad\qquad 	\fnpr\defeq i^{-n(n-1)/2}\gmpr^1\dots \gmpr^n.
\end{equation} 
 
for the inner-products $\langle \, .\, , \fnr\,  .\, \rangle$ and $\langle \, .\, , \fnpr\,  .\, \rangle$ respectively, with $k=n=\dim(T\Man^+)$ as previously and $l=\dim(T\Man^-\cap T\Man_r^+)=0$ since $T\Man^-=T\Man_r^-$ is orthogonal to $T\Man_r^+$. In the same way as before, we define the reflection operators coming with the spacelike reflection $r$ as $\prr(\,.\,)=(-1)^{n+1}\fnr(\,.\,)\fnr$ and $\prpr(\,.\,)=(-1)^{n+1}\fnpr(\,.\,)\fnpr$, offering also a generalized notion of parity operator\footnote{The parity operator is defined in quantum field theory as the operator that reverses the space-like coordinates, i.e. the ones corresponding to negative signatures in the metric.}. 

 \begin{remark}
	The operators $\fnr$ and $\fnpr$ are equal if we make the choice to relate $c$ and $\clpr$ by the relation $\clpr(E_a)=\clr(E_a)$ for $a\in\{1, \dots, n\}$ and $\clpr(E_a)=i\clr(E_a)$ for $a\in\{n+1, \dots, 2m\}$. This is the usual way to implement the signature change, known as Wick rotation.
\end{remark}

These two fundamental symmetries permit to realise the relation of equation \eqref{EqmetDef} or equivalently $\gpr(\, .\,,\, .\, )= \gr(\, .\,,\,r .\, )$ at the level of the Clifford representation. Taking $v, w$ to be vectors in $T\Man$ we have in the two contexts
\begin{itemize}
\item \textbf{Context 1: } starting from $\gr(v, w)\mbb=\frac{1}{2}\{\clr(v), \clr(w)\}$ the metric $\gpr$ is obtained by
\begin{equation}
	\label{EqRelMetC1}
\gpr(v, w)\mbb=\gr(v, rw)\mbb=\frac{1}{2}\{\clr(v), \clr(rw)\}= \frac{1}{2}\{\clr(v), \prr(\clr(w))\}. 
\end{equation}
\item \textbf{Context 2: }starting from $\gpr(v, w)\mbb=\frac{1}{2}\{\clpr(v), \clpr(w)\}$ the metric $\gr$ is obtained by
\begin{equation}
	\label{EqRelMetC2}
\gr(v, w)\mbb=\gpr(v, rw)\bbbone=\frac{1}{2}\{\clpr(v), \clpr(rw)\}=\frac{1}{2}\{\clpr(v), \prpr(\clpr(w))\}. 
\end{equation}
\end{itemize}
The relations of equations \eqref{EqRelMetC1} and \eqref{EqRelMetC2} offer an explicit link between Krein spaces through their fundamental symmetries (in correspondence with spacelike reflection) and the process of signature change, realised here in an algebraic way. These results are an extension of the pioneering work done in \cite{strohmaier2006noncommutative}. 

\subsection{The spectral triples of the connection}
\label{SubSecDiffST}
 
Let's consider the algebra $\calA=C^\infty(\Man)\otimes \mathbb{C}^2$ together with the Hilbert space $\calH=L^2(\Man, \Sp)$ with elements $a=(f,g)\in\calA$ presented according to the chiral representation $\pi_{\,\Gamma\,}$ defined in equation \eqref{EqRho}. The action of $\rho$ on $a$ is given by $\rho(f,g)= (g,f)$. Thanks to proposition \ref{PropInnTw}, we know that $\rho$ is inner. Following \cite{martinetti2024torsion} the twist by grading of the spin manifold can be generated by a unitary action $\rho(a)=\tw a \tw^\dagger$ where $\tw$ can take the general form
\begin{align*}
	\tw=\alpha \prod_{i=1}^{n\,\, odd}\gamma^{a_i}
\end{align*}
with $\alpha$ a complex number such that $\alpha\alpha^*=1$, $a_i\neq a_j$ for $i\neq j$. 

To obtain a Krein space and set up the dualities, we require the twist to be fundamental. We then have $\rho(\, .\,)=\tw(\, .\,)\tw$ where $\tw$ takes the form of equation \eqref{EqCompJ}. The natural question becomes which fundamental symmetries implement the fundamental twist of the connection?
\begin{proposition}
	\label{PropFundCompTw}
Fundamental symmetries $\calJ_r$ of the form of equation \eqref{EqGenFondSym} corresponding to a reflection of an odd number of dimensions are the only ones to implement the (fundamental) twist by the grading.
\end{proposition}
\begin{proof}
Taking an element $a=(f,g)\in\calA$, we have that $\gamma^b a=\rho(a)\gamma^b$ for all $b\in\{1, \dots, 2m\}$ thanks to the relation $\Gamma\gamma^b=-\gamma^b\Gamma$. Then $\calJ_r a =\rho(a)\calJ_r$ if $k$ is odd and $\calJ_r a =a\calJ_r$ if $k$ is even. Then, when $k$ is odd we have
	\begin{equation}
		\calJ_r a \calJ_r=(\calJ_r)^2 \rho(a)= \rho(a). 
	\end{equation}
\end{proof}
In this case, equation \eqref{EqRefleGam} implies that the corresponding twist $\rho(\, .\,)=\calJ_r(\, .\,)\calJ_r$ implements the action of the reflection operator:
\begin{equation}
	\label{EqTwAsRefl}
P_r(\cl(v))=\rho(\cl(v))= \cl(rv).
\end{equation}
From equation \eqref{eqSymCr}, we can now define the two fundamental symmetries, corresponding to a spacelike reflection between $\gr$ and $\gpr$ in each context
\begin{center}
	\begin{minipage}[b]{.41\textwidth}
		\vspace{-\baselineskip}
		\begin{equation}
			\fnr\defeq i^{-n(n-1)/2}\prod_{i=1}^{n\,\, odd}\gmr^{i} \label{EqFondSymCompC1}
		\end{equation}
	\end{minipage}%
	\hfill    \hfill and\hfill
	\begin{minipage}[b]{.44\textwidth}
		\vspace{-\baselineskip}
		\begin{equation}
			\fnpr\defeq i^{-n(n-1)/2}\prod_{i=1}^{n\,\, odd} \gmpr^{i}.	\label{EqFondSymCompC2}
		\end{equation}
	\end{minipage}
\end{center}

From now, every fundamental symmetries will be taken of the form of equations \eqref{EqFondSymCompC1} for context 1 and \eqref{EqFondSymCompC2} for context 2, then compatible with the twist. The compatibility of the space-like reflection's fundamental symmetry with the twist will then impose that we consider only metrics $\gpr$ for which $n=\dim(T\Man^+)$ is odd. We will consider two ways to implement the twist, corresponding to each context:
\begin{itemize}
\item \textbf{Context 1: } $\rhr(\, .\,)=\twr(\, .\,)\twr$ with $\twr=\fnr$ given by equation \eqref{EqFondSymCompC1}.
\item \textbf{Context 2: }$\rhpr(\, .\,)=\twpr(\, .\,)\twpr$ with $\twpr=\fnpr$ given by equation \eqref{EqFondSymCompC2}.
\end{itemize}
 
\begin{remark}
	\label{RqTwOnBH}
Note that $\rhr(a)=\rhpr(a)=\rho(a)$ for any $a\in\calA$ but $\rhr(A)\neq\rhpr(A)$ for $A\in\calB(\calH)$ so the need to introduce these new notations for the action of $\rho$ on all $\calB(\calH)$.
\end{remark}

The twist is then the spacelike reflection operator (or parity operator):
\begin{center}
	\begin{minipage}[b]{.41\textwidth}
		\vspace{-\baselineskip}
		\begin{equation}
			\rhr(\clr(v))=\clr(rv) \label{EqRhPar1}
		\end{equation}
	\end{minipage}%
	\hfill    \hfill and\hfill
	\begin{minipage}[b]{.44\textwidth}
		\vspace{-\baselineskip}
		\begin{equation}
			\rhpr(\clpr(v))=\clpr(rv).	\label{EqRhPar2}
		\end{equation}
	\end{minipage}
\end{center}

Then equations \eqref{EqRelMetC1} and \eqref{EqRelMetC2} became 
\begin{align}
	\label{EqRelMetC1t}
	&\gpr(v, w)\mbb=\gr(v, rw)\mbb= \frac{1}{2}\{\clr(v), \rhr(w)\}\\ 
	\label{EqRelMetC2t}
	&\gr(v, w)\mbb=\gpr(v, rw)\mbb=\frac{1}{2}\{\clpr(v), \rhpr(\clpr(w))\}. 
\end{align}

\begin{remark}
	\label{RqChoixCliff}
 $\Dir$ and $\KDir=\tw\Dir$ cannot be simultaneously associated with a Clifford structure because the product of $\tw$ with the gamma matrices that represent the Clifford structure does not give back a representation of Clifford algebra. Then either $\Dir$ or $\KDir$ can be associated with such a Clifford structure. 
\end{remark}
We will then define new notations in order to indicate whether or not the Dirac operators are associated with Clifford structures. We keep the notations $\Dir$ and $\KDir$ when this is the case and adopt the notations $\HDir$ et $\HTDir$ when they do not have such. The corresponding “modified” version of Clifford structure will be defined later in definition \ref{DefTwCliff}.
 
  \begin{lemma}
  	\label{PropTwSelfgm}
	Taking $\rhpr(\, .\,)=\twpr(\, .\,)\twpr$ with $\twpr=\fnpr$ and $\Tadjpr$ the corresponding adjoint, we then have that
	\begin{equation}
		\label{EqrelatgammaPR}
		(\gmpr^a)^\Tadjpr=\gmpr^a.
	\end{equation} 
\end{lemma}
\begin{proof}
	We have that $(\gmpr^a)^2=\mbb$ if $a\in\{ 1, \dots , n\}$ and $(\gmpr^a)^2=-\mbb$ if $a\in\{n+1, \dots , 2m\}$. Since $\gmpr^a$ are unitary matrices, we have that $(\gmpr^a)^\dagger=(\gmpr^a)^{-1}$ and then that
	\begin{equation}
		(\gmpr^a)^\dagger=\left\{
		\begin{array}{ll}
			\,\, \,\gmpr^a &\text{ for }  \quad a\in\{ 1, \dots , n\},\\ 
			 -\gmpr^a&\text{ for } \quad a\in\{n+1, \dots , 2m\}. \end{array}\right.  
	\end{equation}
	On the other side, from equation \eqref{EqTwAsRefl} we have that 
	\begin{equation}
		\rhpr(\gmpr^a)=\left\{
		\begin{array}{ll}
			\,\, \,\gmpr^a &\text{ for }  \quad a\in\{ 1, \dots , n\},\\ 
			 -\gmpr^a&\text{ for } \quad a\in\{n+1, \dots , 2m\}. \end{array}\right.  
	\end{equation}
 By identification, we have $\rhpr(\gmpr^a)= (\gmpr^a)^\dagger$ which implies $(\gmpr^a)^\Tadjpr=\gmpr^a$ so the result.
\end{proof}
It then follows that the operator $-i\gmpr^a\nabla_a^{{\scriptscriptstyle PR, S}}$ is $\twpr$-selfadjoint so that an adapted notation will be given by $\TprDir=-i\gmpr^a\nabla_a^{{\scriptscriptstyle PR, S}}$. Its dual selfadjoint Dirac operator is given by $\prDir\defeq \twpr \TprDir$. In the same way for context $1$, starting from the selfadjoint Dirac operator $\rDir\defeq-i\gmr^a\nabla_a^{{\scriptscriptstyle R, S}}$ we can define its dual $\twr$-selfadjoint Dirac operator $\HTrDir\defeq \twr \rDir$.

Remarks \ref{RQExchange} and \ref{RqChoixCliff} lead us to distinguish two possibilities:
\begin{itemize}
	\item \textbf{Context 1: }we start from $\rDir$ which is associated with a Riemannian Clifford structure, we have $\Gamma\rDir=-\rDir\Gamma$ with $\rDir\, :\, \calH_\pm\,\to\, \calH_\mp$ and then $\Gamma\HTrDir=\HTrDir\Gamma$ with $\HTrDir\, :\, \calH_\pm\,\to\, \calH_\pm$ for the dual Dirac operator $\HTrDir$.
	\item \textbf{Context 2: }we start from $\TprDir$ associated with a pseudo-Riemannian Clifford structure. We have $\Gamma\TprDir=-\TprDir\Gamma$ with $\TprDir\, :\, \calH_\pm\,\to\, \calH_\mp$ and then $\Gamma\prDir=\prDir\Gamma$ with $\prDir\, :\, \calH_\pm\,\to\, \calH_\pm$ for the dual Dirac operator $\prDir$.
\end{itemize} 

\begin{definition}
	The sub-algebras $\calA^{\prime}$ and $\calA^{\prime\prime}$ of $\calA$ are defined by
	\begin{equation}
		\calA^{\prime}\defeq\{a\in\calA\, / \, \rho(a)=a\}\qquad\quad\text{and}\quad\qquad 	\calA^{\prime\prime}\defeq\{a\in\calA\, / \, \rho(a)\neq a\}
	\end{equation}
	which is equivalently given by the set of elements $a\in\calA$ of the form $a=(f,f)$ for $a\in \calA^{\prime}$ and by the set of elements $a=(f,g\neq f)$ for $a\in \calA^{\prime\prime}$.
\end{definition}

\begin{lemma}
	\label{PropComutBorn}
Taking $\rDir=-i\gmr^a\nabla_a^{{\scriptscriptstyle R, S}}$ and $\TprDir=-i\gmpr^a\nabla_a^{{\scriptscriptstyle PR, S}}$ with corresponding Dirac operators $\HTrDir=\twr \rDir$ and $\prDir= \twpr \TprDir$ then we have that the commutators $[\prDir, a]$, $[\HTrDir, a]$, $[\TprDir, a]_\rho$ and $[\rDir, a]_\rho$ are bounded on $\calA$ and that the commutators $[\prDir, a]_\rho$, $[\HTrDir, a]_\rho$, $[\TprDir, a]$ and $[\rDir, a]$ are unbounded on $\calA^{\prime\prime}$ and bounded on $\calA^{\prime}$.
\end{lemma}
\begin{proof}
Using the fact that $\forall a\in\calA$ and $b\in\{1, \dots, 2m\}$ we have that $\gmr^b$ and $\gmpr^b$ twist commute with the algebra, i.e. $[\gmr^b, a]_\rho=[\gmpr^b, a]_\rho=0$ and taking $\psi\in\calH$, we have that $\rDir(a\psi)=\rho(a)\rDir(\psi)+\rDir(a)\psi$ and $\TprDir(a\psi)=\rho(a)\TprDir(\psi)+\TprDir(a)\psi$. The computation of the different commutators gives
\begin{align}
	\label{EqCom1}
&[\rDir, a]=\twr[\HTrDir, a]_\rhr=\rDir(a)+(\rho(a)-a)\rDir\\
\label{EqCom2}
&[\TprDir, a]=\twpr[\prDir, a]_\rhpr=\TprDir(a)+(\rho(a)-a)\TprDir\\
\label{EqCom3}
& [\rDir, a]_\rhr=\twr[\HTrDir, a]=  \rDir(a).\\
\label{EqCom4}
&[\TprDir, a]_\rhpr=\twpr [\prDir, a]= \TprDir(a)
\end{align}
We directly observe that the commutators $[\TprDir, a]_\rhpr$ (then $[\prDir, a]$) and $[\rDir, a]_\rhr$ (then $[\HTrDir, a]$) are bounded on all $\calA$. Since $\TprDir$ and $\rDir$ are not bounded, the commutators $[\rDir, a]$ (then $[\HTrDir, a]_\rhr$) and $[\TprDir, a]$ (then $[\prDir, a]_\rhpr$) are bounded only when the twist is trivial for the chosen algebra, then when $a\in \calA^{\prime}$.
\end{proof}

\begin{remark}
This is interesting to note that doubling the algebra $C^\infty(\Man)$ according to the chiral spaces leads to the fact that $[\rDir, a]$ and $[\TprDir, a]$ are no longer bounded and must therefore be replaced by the twisted commutators $[\rDir, a]_\rhr$ and $[\TprDir, a]_\rhpr$ in order to define a spectral triple, as mentioned in the introduction.
\end{remark}
 
Taking $\calH=L^2(\Man, \Sp)$ and $\calK_\tw$ being equivalent to $L^2(\Man, \Sp)$ where $\langle \, .\, ,\, .\, \rangle$ is replaced by $\langle \, .\, ,\, .\, \rangle_\tw$ and using lemma \ref{PropComutBorn}, we can now define the generalized spectral triples on $\calA^\prime$ with the index corresponding to the context (1 or 2) taken to define the spectral triple:
\begin{itemize}
\item  The $ST_{(1)}\defeq( \calA^{\prime} , \calH,  \rDir , J, \Gamma)$ and $\twpr$-$PRST\defeq(\calA^{\prime}, \calK_\twpr,  \TprDir, J, \Gamma)$.
\item The $\twpr$-$TST\defeq(\calA^{\prime}, \calH,  \prDir, J, \Gamma, \twpr)$ and $\twr$-$TPRST\defeq(\calA^{\prime}, \calK_\twr,  \HTrDir, J, \Gamma, \twr)$. 
\end{itemize}
and the 4 other ones on all $\calA$:
\begin{itemize}
\item  The $ST_{(2)}\defeq( \calA , \calH,  \prDir , J, \Gamma)$ and $\twr$-$PRST\defeq(\calA, \calK_\twr,  \HTrDir, J, \Gamma)$.
\item The $\twr$-$TST\defeq(\calA, \calH,  \rDir, J, \Gamma, \twr)$ and $\twpr$-$TPRST\defeq(\calA, \calK_\twpr,  \TprDir, J, \Gamma, \twpr)$. 
\end{itemize}
The chiral representation (equation \eqref{EqRho}) and the corresponding action of the twist only depend on the choice of $\calA$ and $\Gamma$, then are the same in each context.

Then connection 1 extends to two dualities according to the chosen context, i.e. connection 1.1 for context 1 (resp 1.2 for context 2):
\begin{itemize}
	\item \textbf{connection 1.1:} between $ST_{(1)}$ and $\twr$-$TPRST$.
	\item \textbf{connection 1.2:} between $ST_{(2)}$ and $\twpr$-$TPRST$.
\end{itemize}
 And similarly for connection 2:
\begin{itemize}
	\item \textbf{connection 2.1:} between $\twr$-$PRST$ and $\twr$-$TST$.
	\item \textbf{connection 2.2:} between $\twpr$-$PRST$ and $\twpr$-$TST$.
\end{itemize} 
\subsection{Signature change by the $\tw$-morphism}
\label{SubSectDualMet}

We can now show how the $\tw$-morphism between the previous generalized spectral triples implements a signature change from $\gr$ to $\gpr$, as  presented in subsection \ref{SubsecWickPar}.

\begin{lemma}[Useful relations]
	Taking $a=(f, g)$ and ${\bar a}=(\bar f, \bar g)$ its complex conjugate and noting that $i[\rDir, a]=\clr(d(a))$ and $i[\TprDir, a]=\clpr(d(a))$, we have $\forall a\in \calA$:
	\begin{align}
		\label{EqRelclradj}
		&\clr(da)^\dagger=\clr(d(\rho({\bar a})))\\
		\label{EqRelclpradj}
		&\clpr(da)^\dagger=\clpr(rd({\bar a}))
	\end{align}
\end{lemma}
\begin{proof}
	Concerning equation \eqref{EqRelclradj} we have that
	\begin{equation}
		\clr(da)^\dagger=\partial_b({\bar a})(\gmr^b)^\dagger=\partial_b({\bar a})\gmr^b=\gmr^b\rho(\partial_b({\bar a}))=\gmr^b\partial_b(\rho({\bar a}))=\clr(d(\rho({\bar a})))
	\end{equation}
	where we use that $\gmr^b a=\rho(a)\gmr^b$ for any $a\in\calA$.
	
For equation \eqref{EqRelclpradj}
	\begin{equation}
		\clpr(da)^\dagger=\partial_b({\bar a})\rhpr(\gmpr^b)=\rhpr(\gmpr^b)\rhpr(\partial_b({\bar a}))=\rhpr(\clpr(d({\bar a})))=\clpr(rd({\bar a}))
	\end{equation}
	where we use that $(\gmpr^a)^\dagger=\rhpr(\gmpr^a)$ (see lemma \ref{PropTwSelfgm}) and equation \eqref{EqRhPar2}.
\end{proof}

In the usual case, for the $ST_{(1)}$, the distance formula is provided by
\begin{align}
	\label{MetForm}
	d_\gr(x,y)=\sup_{a\in \calA^\prime}\{ |a(x)-a(y)|\, \mid\, \| [\rDir, a]\|\leq 1 \}.
\end{align}
The metric information is hidden in the term $\| [\rDir, a]\|$, i.e., for any $a\in\calA^\prime$
\begin{align*}
	\| [\rDir, a]\|^2&=\frac{1}{2}\left(\| \clr(da)\|^2+ \| \clr(da)^\dagger\|^2 \right)\\
	&=\frac{1}{2}\sup_{\substack{\psi\in\calH\\  \, \|\psi \| =1}}\{ \langle \left(\clr(d(\rho({\bar a})))\clr(da)+\clr(da)\clr(d(\rho({\bar a}))) \right)\psi , \psi \rangle  \}\\
	&=\sup_{\substack{\psi\in\calH\\  \, \|\psi \| =1}}\{ \langle \gr(grad \, a, grad \, {\bar a})\psi , \psi \rangle  \}= \|grad\, a\|_\infty^2
\end{align*}
using equation \eqref{EqRelclradj} at the second line, the relation between the Clifford action and the metric and that $\rho(a)=a$ for $a\in\calA^\prime$. The last term is the Lipschitz norm of $a$ for the metric $\gr$, see \cite{gracia2013elements} for more information. Note that the self-adjointness of $\clr(da)$ is an essential property to make the sum of supremums a supremum of the sum, as this implies that the supremum is attained for the same vector $\psi$ in the expression of $\| \clr(da)\|^2$ and $\| \clr(da)^\dagger\|^2$.

There is no such distance formula in the compact pseudo-Riemannian spectral triple framework. In particular, our use of the distance formula \eqref{MetForm} will always refer to the underlying Riemannian metric $\gr$ encoded in the $ST_{(1)}$ and its twisted counterparts. For the $\twpr$-$PRST$ and the $\twpr$-$TPRST$, the only metric formula is provided by the Clifford relation connecting $\clpr$ to $\gpr$.

In the case of connection 2.2, the $\tw$-morphism permits the transformation $\twphi\, :\, \twpr\text{-}PRST\, \to\, \twpr\text{-}TST$, but so far it is not clear which metric the $\twpr$-$TST$ can be associated with. The following important theorem answers this question.
\begin{theorem}
	\label{DistanceTW}
The $\twpr$-$TST$ is associated with the metric $\gr$.
\end{theorem}
\begin{proof}
Using the distance formula \eqref{MetForm}, the only difference in the $\twpr$-$TST$ appears at the level of the derivation, giving
\begin{align}
	d(x,y)=\sup_{a\in \calA^\prime}\{ |a(x)-a(y)|\, \mid\, \| [\prDir, a]_\rhpr  \|\leq 1 \}.
\end{align}
Computing $\| [\prDir, a]_\rhpr  \|^2$ yields
\begin{align*}
	\| [\prDir, a]_\rhpr  \|^2&=\|\twpr\clpr(da)\|^2=\frac{1}{2}\left(\|\clpr(da)\|^2+ \| \clpr(da)^\dagger \|^2 \right)\\	&=\frac{1}{2}\sup_{\substack{\psi\in\calH\\  \, \|\psi \| =1}}\{ \langle \left(\clpr(rd({\bar a}))\clpr(da)+\clpr(da)\clpr(rd({\bar a}))\right)\psi , \psi \rangle  \}\\
	&=\sup_{\substack{\psi\in\calH\\  \, \|\psi \| =1}}\{ \langle \left(\gpr(grad\, a, r grad\, {\bar a}) \right)\psi , \psi \rangle  \}\\
	&=\sup_{\substack{\psi\in\calH\\  \, \|\psi \| =1}}\{ \langle \left(\gr(grad\, a, grad\, {\bar a}) \right)\psi , \psi \rangle  \}= \|grad\, a\|_\infty^2.
\end{align*}
Using equation \eqref{EqRelclpradj} in the second line. This implies that $d(x,y)=d_\gr(x,y)$, so that the metric associated with the $\twpr$-$TST$ is $\gr$.
\end{proof}

For the connection 1.2, the computation (hence the result) for the distance formula of the $ST_{(2)}$ is the same for the $\twpr$-$TST$ defined on $\calA$, implying that the $ST_{(2)}$ can be associated with the metric $\gr$. However for connection 2.1, a difference in the distance formula appears for the $\twr$-$TST$ on $\calA$ as we obtain that 
\begin{align*}
	\| [\rDir, a]_\rhr  \|^2&=\sup_{\substack{\psi\in\calH\\  \, \|\psi \| =1}}\{ \langle \left(\gr(grad\, a, grad\, {\rho(\bar a)}) \right)\psi , \psi \rangle  \}.
\end{align*}
It is not clear in this case whether the distance formula can be associated with the metric $\gr$. We can nevertheless impose a Clifford relation to connect $\clr$ to $\gr$ to provide a metric formula for the $\twr$-$TST$ or to restrict to $\calA^\prime$.

From these arguments, the action of the $\tw$-morphism can be regarded as implementing a change of signature at the level of the metric data only for connections 1.2 and 2.2. The effect on the metric remains unclear for connections 1.1 and 2.1, since in our compact pseudo-Riemannian setting we do not dispose of a Connes-type distance formula for the $\twr$-$TPRST$ and the $\twr$-$PRST$. Let us stress, however, that spectral formulations of the Lorentzian distance do exist in the framework of Lorentzian spectral triples on non-compact, stably causal space-times; see for instance \cite{franco2014temporal,franco2018lorentzian}. These constructions rely on non-unital algebras and global hyperbolicity assumptions, and are therefore not directly applicable to the compact pseudo-Riemannian triples considered in the present work. Extending our twist construction to such globally hyperbolic Lorentzian spectral triples lies beyond the scope of this paper.

A general formula for the local extraction of the metrics is provided by the following
\begin{align}
	\label{MetFormR}
\gr(v, w)&=\frac{1}{2^m}\Tr(\clr(v)\clr(w)),\qquad\quad \,\,\,\text{for the $ST_{(1)}$ and the $\twr$-$TST$}.\\
\label{MetFormPR}
\gpr(v, w)&=\frac{1}{2^m}\Tr(\clpr(v)\clpr(w)),\qquad\quad \text{for the $\twpr$-$PRST$ and the $\twpr$-$TPRST$}.
\end{align}
where $v,w\in T_x\Man$ for a given point $x\in \Man$.

The $\tw$-morphism action on the Dirac operators induces an action on the Clifford representation given by
\begin{align*}
\cl(v)\xrightarrow{\twphi}\tw \cl(v).
\end{align*}
This action extends to an action on the metric formulas in \eqref{MetFormR} and \eqref{MetFormPR}
\begin{align}
	\label{EqAZE}
	&\gr(v, w)=\frac{1}{2^m}\Tr(\clr(v)\clr(w))\,\xrightarrow{\twphir}\, \frac{1}{2^m}\Tr(\twr\clr(v)\twr\clr(w)),\\
	\label{EqAZD}
&\gpr(v, w)=\frac{1}{2^m}\Tr(\clpr(v)\clpr(w))\,\xrightarrow{\twphipr}\, \frac{1}{2^m}\Tr(\twpr\clpr(v)\twpr\clpr(w)).
\end{align}
Using \eqref{EqRhPar1} and \eqref{EqRhPar2}, equations \eqref{EqAZE} and \eqref{EqAZD} became 
\begin{align}
 &\frac{1}{2^m}\Tr(\twr\clr(v)\twr\clr(w))=\frac{1}{2^m}\Tr(\clr(rv)\clr(w))=\gr(rv, w)=\gpr(v,w),\\
	&\frac{1}{2^m}\Tr(\twpr\clpr(v)\twpr\clpr(w))=\frac{1}{2^m}\Tr(\clpr(rv)\clpr(w))=\gpr(rv, w)=\gr(v,w).
\end{align}
The $\tw$-morphism action then extends to an action on the metric given by $\gr\,\xleftrightarrow{\twphi}\,\gpr$, then to a signature change.
 
It is important to emphasize that all the operations producing the signature change are defined locally on $\Man$. In particular, when $\gpr$ has Lorentzian signature, our compactness assumption prevents $(\Man,\gpr)$ from being a globally hyperbolic space-time, so that no global causal structure is reconstructed from our data. The signature change implemented by the $\tw$-morphism should therefore be understood as a local modification of the metric, in the same spirit as the Wick rotation in quantum field theory. Whether and how such a local construction can be extended to a non-compact globally hyperbolic Lorentzian manifold is a separate problem that we do not address in this work.

The following table summarizes the results obtained in this subsection.

\begin{table}[H]
	\centering
	\begin{tabular}{|P{1.9cm}|P{1.3cm}|P{4.7cm}|P{5.3cm}|}
		
		\hline
		connection & algebra & spectral triple and metric & dual spectral triple and metric \\ \hline
		\multirow{1}{*}{1.1}
		& $\calA^{\prime}$ & $ST_{(1)}\,\,$ and  $\,\,\gr(\, .\, ,\, .\,)$ & $\twr$-$TPRST\,\,$ and  $\,\,\gr(\, .\, , r\, .\,)$  \\
		\hline
		\multirow{1}{*}{2.1}
		& $\calA$ &$\twr$-$TST\,\,$ and  $\,\,\gr(\, .\, ,\, .\,)$ & $\twr$-$PRST\,\,$ and  $\,\,\gr(\, .\, , r\, .\,)$ \\
		\hline
		\multirow{1}{*}{2.2}
		& $\calA^{\prime}$ & $\twpr$-$TST\,\,$ and  $\,\,\gpr(\, .\, , r \, .\,)$ & $\twpr$-$PRST\,\,$ and  $\,\,\gpr(\, .\, ,\, .\,)$ \\
		\hline
		\multirow{1}{*}{1.2}
		& $\calA$ & $ST_{(2)}\,\,$ and  $\,\,\gpr(\, .\, , r \, .\,)$ & $\twpr$-$TPRST\,\,$ and  $\,\,\gpr(\, .\, ,\, .\,)$ \\
		\hline
	\end{tabular}
	\caption{\label{TableResMet} Relation between metrics for each connection.}
\end{table}
The signature change is given by the passage from one column of spectral triples to the one of the dual spectral triple, implemented by the corresponding $\tw$-morphisms. In context 1 (dualities 1.1 and 2.1), the transition between columns is governed by the action of $\phi^{\,\twr}$ or, equivalently, by the transformation $\twr\, : \mbb\,\to\, \fnr$. Similarly, in context 2 (dualities 2.2 and 1.2), the transition between columns is carried out by $\phi^{\,\twpr}$ or, equivalently, by the transformation $\twpr\, : \mbb\,\to\, \fnpr$.

All of these results motivate the following definitions, to construct an analogue of the Clifford algebraic structure for the Dirac operators $\HDir$.
 \begin{definition}[Twisted Clifford algebra]
 	\label{DefTwCliff}
Taking a finite dimensional vector space $V$ together with a nondegenerate symmetric bilinear form $g$, the $\rho$-twisted Clifford algebra $Cl_\rho(V, g)$ is the finite dimensional algebra generated by the elements $c(x)$ where $x\in V$ subject to the relation
 	\begin{equation}
 		\{c(x), \rho(c(y))\}=2g(x, y)\bbbone \qquad\qquad \forall x,y \in V
 	\end{equation}
 where $\rho$ is a regular automorphism of $Cl_\rho(V, g)$.
 \end{definition}
An example is given by the case where the vector space is the tangent space $T\Man$, starting from the previous Clifford algebra $Cl(T\Man, g)$ where we have $\{c(v),c(w)\}=2g(v,w)\bbbone_{2^{m}}$. Taking the fundamental twist $\rho(\,.\,)=\tw(\,.\,)\tw$ with $\tw=\calJ_r$ of the form of equation \eqref{EqGenFondSym}, we define the corresponding $\rho$-twisted Clifford algebra $Cl_\rho(T\Man, g^r)$ as the algebra generated by the elements $\cl(v)$ with relation
\begin{equation}
	\{c(v), \rho(c(w))\}=2g^r(v, w)\bbbone_{2^{m}}
\end{equation}
where $r$ is defined from the twist by the relation $\rho(c(v))=c(rv)$. The metrics $g^r$ and $g$ are related by $g(v, r w)=g^r(v, w)$.
\begin{definition}[Twisted Dirac operator]
Starting from a $\rho$-twisted Clifford algebra as defined previously with $\rho(\,.\,)=\tw(\,.\,)\tw$. A $\tw$-twisted Dirac operator is defined by
\begin{equation}
	\label{EqRhoDir}
\HDir=-i\tw c(\theta^a)\nabla_a^S
\end{equation}
where $\nabla_a^S$ is the spin connection associated with $\g$.
\end{definition}

As shown in Table \ref{TableResMet}, the $\tw$-twisted Dirac operator naturally arises from the connection and may be linked to the $\rho$-twisted Clifford algebra. Equation \eqref{EqRhoDir} expresses $\HDir$ as $\HDir = \tw \Dir$, where $\Dir = -ic(\theta^a)\nabla_a^S$ corresponds to the metric $\g$ and Clifford algebra $Cl(T\Man, g)$, while $\HDir$ is associated with $\g^r$ and the induced $\rho$-twisted Clifford algebra $Cl_\rho(T\Man, g^r)$. The transformation $\tw : \mbb \to \calJ_r$, implemented via the $\tw$-morphism, maps a Clifford algebra to its $\rho$-twisted counterpart, enabling signature changes when they occur at the Clifford structure level. These results apply to even-dimensional manifolds. Since the twist is fundamental, Proposition \ref{PropFundCompTw} ensures that the signature change affects an odd number of metric entries in an orthonormal basis, encompassing the crucial transition from Euclidean to Lorentzian manifolds, the core of the signature problem. It is noteworthy that Proposition~\ref{PropFundCompTw} connects the twist by the grading with an odd number of sign changes; whether this reflects a deeper underlying structure remains an open question. A more thorough characterization of the twisted Clifford algebra and the twisted Dirac operator will be undertaken in future work.

\subsection{The 4-dimensional compact Lorentzian model}
\label{SubSec4dlor}
 
An important example of physical interest is provided by a 4-dimensional compact pseudo-Riemannian manifold $(\Man,\gpr)$ with metric signature $(1,3)$. Throughout this subsection we still assume that $\Man$ is compact, so $(\Man,\gpr)$ should be regarded as a local Lorentzian model for the Dirac operator rather than as a realistic globally hyperbolic space-time. The metric $\gpr$ admits a local oriented orthonormal basis $\{E_1,\dots,E_4\}$ with a splitting $T\Man = T\Man^+\oplus T\Man^-$ induced by $\gpr$, such that $\dim(T\Man^+)=1$ and $\dim(T\Man^-)=3$, where $T\Man^+$ is timelike and $T\Man^-$ spacelike.

 Then $\gpr$ admits a local oriented orthonormal basis $\{E_1, \dots, E_4\}$ with splitting $T\Man=T\Man^+\oplus T\Man^-$ induced by $\gpr$ so that $\dim(T\Man^+)=1$ (then $\dim(T\Man^-)=3$) with $\{E_1\}$ the basis of $T\Man^+$ and $\{E_2, E_3, E_4\}$ the one of $T\Man^-$. As previously, the physical Dirac operator is given by $\TprDir=-i\gmpr^a\nabla_a^{{\scriptscriptstyle PR, S}}$. A spacelike reflection $r$ can then be defined so that $rv^\pm=\pm v^\pm$ for $v^\pm\in T_x\Man^\pm$. Following equation \eqref{EqFondSymCompC2} the associated fundamental symmetry is given by $\fnpr=\gmpr^1$. The corresponding twist is given by the action $\rhpr(\, .\,)=\twpr(\, .\,)\twpr$ with $\twpr=\fnpr$ and a selfadjoint Dirac operator $\prDir$ is deduced from $\TprDir$ by the relation $\prDir= \twpr \TprDir$.

We are then left with these two options:
\begin{itemize}
\item \textbf{connection 1.2} between the corresponding $ST_{(2)}$ and $\twpr$-$TPRST$ on $\calA$.
\item \textbf{connection 2.2} between the corresponding $\twpr$-$PRST$ and $\twpr$-$TST$ on $\calA^\prime$. 
\end{itemize}
We have from corollary \ref{PropDualFermAct} that the physical Dirac action of quantum field theory is obtained in each case by the evaluation of the corresponding fermionic actions:
\begin{equation}
\act_f^\twpr( \TprDir, \psi)=\langle \psi , \TprDir\psi \rangle_\twpr=  \langle \psi , \prDir\psi \rangle=  \act_f( \prDir, \psi)
\end{equation} 
where $\act_f^\twpr$ (resp $\act_f$) corresponds to the fermionic action used for the $\twpr$-$TPRST$ and the $\twpr$-$PRST$ (resp for the $ST_{(2)}$ and the $\twpr$-$TST$).
 
Moreover, particular examples of $\twpr$-unitary operators are given by Lorentz transformations, when taking the twist to be implemented by the first gamma matrix $\gmpr^1$ i.e. $\rho(\, .\,)=\gmpr^1(\, .\,)\gmpr^1$ (see \cite{martinetti2024torsion}). This fermionic action is naturally preserved by $\twpr$-unitaries:
\begin{equation}
	\act_f^\twpr( U_\twpr\TprDir U_\twpr^\Tadjpr, U_\twpr\psi)=\act_f^\twpr( \TprDir, \psi)\qquad\text{with}\qquad U_\twpr\in \calU_\twpr(\calB(\calH))
\end{equation} 
and then by Lorentz transformations. 
\begin{remark}
Such a Dirac operator $\prDir$ has also been used in the name of Krein-shifted Dirac operator (see \cite{bochniak2020spectral, bochniak2019pseudo}) as an attempt to solve the fermion doubling problem. 
\end{remark} 

Starting from context 2 is physically justified as it provides a Lorentzian Dirac operator, suggesting that $ST_{(2)}$ or $\twpr$-$TST$ may be suitable spectral triple candidates for the almost-commutative manifold in the noncommutative Standard Model. Conversely, starting from context 1 allows deducing the Lorentzian Dirac operator $\HTrDir$, dual to the usual Riemannian $\rDir$. Notably, in this setting, for the $\twr$-$TST$ on a 4D manifold, a geodesic-preserving torsion emerges (see Section \ref{SecTwsp} and \cite{martinetti2024torsion}), taking the form of axial terms proportional to the grading $\Gamma$, known in QFT for breaking parity symmetry.

\newpage

\section{Summary and outlook}

We established a precise correspondence between the twist-based and pseudo-Riemannian approaches to noncommutative geometry. The central result, Theorem~\ref{Thmconnection}, shows that the four generalized spectral triples of Definition~\ref{Def4kindsST} are related by two bijective $\tw$-morphisms preserving the algebraic axioms and the structure of gauge fluctuations. Corollary~\ref{PropDualFermAct} further demonstrates that both the bosonic and fermionic actions are stable under this correspondence. On the geometric side, Theorem~\ref{DistanceTW} proves that the $\tw$-morphism acts as a local signature change transformation in the case of connection 2, providing an alternative to Wick rotation. This approach to implement signature change appears more natural than usual Wick rotation, as it applies uniformly to the algebraic structure of each geometric dimension via a single operator $\tw$, which, in quantum field theory, is explicitly given by the temporal gamma matrix. This provide a unified algebraic mechanism for relating Euclidean and Lorentzian data locally, for compact manifolds.

It may retrospectively appear more natural to require $\rho$ to be $\calB(\calH)$-regular (instead of being regular only on the algebra) in the axiom of the twisted spectral triple since the twist $\rho$ induces a $\tw$-inner product defined on all of $\calH$. Moreover, the specificity of their corresponding fundamental symmetries plays a key role in implementing the signature change. In this case, taking $\rho\in\Inn(\calB(\calH))$ with $\rho(\, .\,)=\tw(\, .\,) \tw^\dagger$, lemma \ref{PropRegImpFund} ensures that $\tw=\exp(i\theta)\tw^\dagger$, so that the $\tw$-inner product is a Krein product only if $\theta=2n\pi$ for some integer $n$. The condition $\tw=\tw^\dagger$ emerges as a fundamental property of the connection. It ensures the self-adjointness of the Dirac operators with respect to their inner products, establishes the relation between definite and indefinite inner products (see Remark \ref{RqTwRelatDef}), and ultimately to set up the connection. 

The special case of even-dimensional manifolds is particularly interesting. Observing that fundamental symmetries are linked to reflections (see Proposition \ref{PropParOp}), we demonstrate that in the case of space-like reflections of an odd number of dimensions, the corresponding twist is the parity operator, which enables the definition of signature change. The notion of twisted Clifford algebra (see Definition \ref{DefTwCliff}) provides a framework in which the Dirac operators denoted by $\HDir$ (which are not associated with a Clifford structure) acquire meaning, along with their signature-transformed metrics.

 An interesting feature of $\tw$-morphisms is that they are fully parametrized by the operators $\tw$, presenting the corresponding signature changes as "deformations" of spectral triples that preserve the fermionic and spectral actions.   

This raises the question of why physics appears to favour one of the two dual spectral triples. The fermionic actions $\act_f^\twpr(\TprDir,\psi)$ and $\act_f(\prDir,\psi)$ differ in the symmetry groups they admit, since they arise from distinct inner products. Given the Lorentzian nature of physical space-time, only $\act_f^\twpr(\TprDir,\psi)$, and thus the $\twpr$-$PRST$, is compatible with a Lorentz-invariant scalar mass term. Combined with the fact that, in the noncommutative Standard Model, masses originate from the finite part of the almost commutative geometry, this suggests that Lorentz symmetry constrains the choice of spectral triple entering the connection and points to a relation between the finite noncommutative sector and the emergence of the time direction. We hope that this connection will provide a concrete path towards a Lorentzian formulation of almost commutative geometries, for instance by incorporating one of the dual spectral triples of Subsection~\ref{SubSec4dlor} into the noncommutative Standard Model. The extension of this programme to non-compact globally hyperbolic space-times is left for future work.

\newpage
\section*{Acknowledgements}
This work is part of the project “Geometry of fundamental interactions: from quantum space to quantum spacetime” funded by the CARIGE foundation. 

%
%

\bibliography{bibliography}

\end{document}